%% file: main.tex
\documentclass[10pt]{article}
\input{macros.tex}

\begin{document}

\title{A Survey on the Densest Subgraph Problem and Its Variants}

\author[1,3]{Tommaso Lanciano}
\author[2]{Atsushi Miyauchi}
\author[2]{Adriano Fazzone}
\author[2]{Francesco Bonchi\thanks{Contact: \href{mailto:bonchi@centai.eu}{bonchi@centai.eu}}}
\affil[1]{KTH Royal Institute of Technology, Sweden}
\affil[2]{CENTAI Institute, Italy}
\affil[3]{Sapienza University of Rome, Italy}
\date{}
\maketitle \sloppy

\begin{abstract}
The Densest Subgraph Problem requires to find, in a given graph, a subset of vertices whose induced subgraph maximizes a measure of density. The problem has received a great deal of attention in the algorithmic literature since the early 1970s, with many variants proposed and many applications built on top of this basic definition. Recent years have witnessed a revival of research interest \rev{in} this problem with several \rev{important} contributions, including some groundbreaking results, published in 2022 and 2023.
This survey provides a deep overview of the fundamental results and an exhaustive coverage of the many variants proposed in the literature, with a special attention \rev{to} the most recent results. The survey also presents a comprehensive overview of applications and discusses some interesting open problems for this evergreen research topic.
\end{abstract}

\pagebreak

\tableofcontents

\pagebreak

\section{Introduction}
\label{sec:intro}
\input{intro.tex}

\section{Fundamental Results}
\label{sec:fundamental}
\input{fundamental.tex}

\section{Adding Constraints}
\label{sec:cons}
\input{cons.tex}

\section{Changing the Objective Function}
\label{sec:variants}
\input{variants.tex}

\section{Richer Graphs}
\label{sec:graphs}
\input{graphs.tex}

\section{Streaming and Distributed}
\label{sec:computational}

\input{computational.tex}

\section{Applications}
\label{sec:apps}
\input{apps.tex}

\section{Conclusions and Open Problems}
\label{sec:conclusions}
\input{conclusions.tex}

\bibliographystyle{acm}
\bibliography{biblio}
\end{document}

%% file: macros.tex
\usepackage{fullpage}
\usepackage{graphicx}
\usepackage{amsmath}
\usepackage{amssymb}
\usepackage{epsfig}
\usepackage{float}
\usepackage{bm}
\usepackage{authblk}

\usepackage[linesnumbered, ruled, vlined]{algorithm2e}
\SetKwInOut{Input}{Input}
\SetKwInOut{Output}{Output}

\usepackage{esvect}
\usepackage{cite}

\usepackage{url}
\usepackage{hyperref}
\usepackage{bbm}
\usepackage[export]{adjustbox}
\usepackage{xspace}
\usepackage{booktabs}
\usepackage{tikz}
\usetikzlibrary{shadows}
\usepackage{xcolor}
\usepackage{amsthm}
\usepackage{cancel}
\usepackage{subfigure}
\usepackage[normalem]{ulem}
\usepackage{multirow}
\usepackage{soul}

\hypersetup{
    bookmarks=true,         
    unicode=false,          
    pdftoolbar=true,        
    pdfmenubar=true,        
    pdffitwindow=false,     
    pdfstartview={FitH},    
    pdftitle={My title},    
    pdfauthor={Author},     
    pdfsubject={Subject},   
    pdfcreator={Creator},   
    pdfproducer={Producer}, 
    pdfnewwindow=true,      
    colorlinks=true,       
    linkcolor=black,          
    citecolor=blue,        
    filecolor=black,      
    urlcolor=blue           
}

\newcommand{\spara}[1]{\smallskip\noindent{\bf #1}}

\newtheorem{theorem}{Theorem}

\newtheorem{lemma}{Lemma}

\DeclareMathOperator*{\argmin}{arg\,min}
\DeclareMathOperator*{\argmax}{arg\,max}
\newcommand{\bigO}{O\xspace}

\newcommand{\NP}{\ensuremath{\mathrm{NP}}\xspace}

\newcommand{\refalg}[1]{Algorithm~\ref{alg:#1}}

\newcommand{\refsec}[1]{Section~\ref{sec:#1}}

\newcommand{\hide}[1]{}

\newcommand{\squishlist}{
 \begin{list}{$\bullet$}
  {  \setlength{\itemsep}{0pt}
     \setlength{\parsep}{3pt}
     \setlength{\topsep}{3pt}
     \setlength{\partopsep}{0pt}
     \setlength{\leftmargin}{2em}
     \setlength{\labelwidth}{1.5em}
     \setlength{\labelsep}{0.5em}
} }
\newcommand{\squishlisttight}{
 \begin{list}{$\bullet$}
  { \setlength{\itemsep}{0pt}
    \setlength{\parsep}{0pt}
    \setlength{\topsep}{0pt}
    \setlength{\partopsep}{0pt}
    \setlength{\leftmargin}{2em}
    \setlength{\labelwidth}{1.5em}
    \setlength{\labelsep}{0.5em}
} }

\newcommand{\squishdesc}{
 \begin{list}{}
  {  \setlength{\itemsep}{0pt}
     \setlength{\parsep}{3pt}
     \setlength{\topsep}{3pt}
     \setlength{\partopsep}{0pt}
     \setlength{\leftmargin}{1em}
     \setlength{\labelwidth}{1.5em}
     \setlength{\labelsep}{0.5em}
} }

\newcommand{\squishend}{
  \end{list}
}

\newcommand{\poly}{\operatorname{poly}}

\newcommand{\rev}[1]{{\color{black}#1}}

%% file: intro.tex
Extracting a dense subgraph from a given input network is a key primitive~\cite{aggarwal,gionis2015dense} in graph mining which, while being investigated since the seventies, keeps attracting a lot of algorithmic research attention nowadays.
Depending on the semantics of the given network, a dense subgraph can represent different interesting patterns: for instance, it can represent a community of closely connected individuals in social networks, a regulatory motif in genomic DNA, or a set of actors coordinating a financial fraudulent activity in a transaction network (see Section \ref{sec:apps} for applications).

\rev{While the literature about extracting dense substructures from graphs is much wider, the focus of the present survey is exclusively on the classic \emph{Densest Subgraph Problem} (DSP) and its many variants. Here, we acknowledge the existence of the literature on related notions, which are not the focus of this survey.}
For instance, the substantial algorithmic effort has been devoted to enumerating all existing dense structures, such as \emph{maximal cliques}, \emph{quasi-cliques}, \rev{\emph{paracliques}, $s$-\emph{cliques}}, $k$-\emph{plex}, $k$-\emph{club}, etc.,~\cite{aggarwal,WU2015693,Chesler+06,chang2018cohesive,farago2019survey,fang2022cohesive,Shahinpour+13}. As we will see more in details later, also the notion of \emph{core decomposition} (see a survey in \cite{malliaros2020core}) is strongly related.
Finding ``clusters'' in graphs is also typically operationalized as finding groups of vertices which are densely connected inside and sparsely connected with the outside. In this regards, \emph{graph clustering}, \emph{graph partitioning}, \rev{\emph{cluster editing},} \emph{spectral clustering} \cite{schaeffer2007graph,malliaros2013clustering,nascimento2011spectral,bulucc2016recent,Bocker+13} are all related notions, as it is the popular topic of \emph{community detection} \cite{fortunato2010community}, or the notion of \emph{correlation clustering} \cite{bonchi_cc}.

Departing from the literature, in this survey we focus on the fundamental problem of extracting \emph{one and only one subgraph}, that maximizes a measure of density.
In the more general setting, we are given a simple graph $G=(V,E)$ \rev{with $n=|V|$ vertices and $m=|E|$ edges. Given} a subset of vertices $S\subseteq V$,  let $G[S]=(S,E[S])$ be the subgraph induced by $S$, and let $e[S]$ be the size of $E[S]$.
The most straightforward notion of density is the so-called \emph{edge density}, defined for a set of vertices $S$ as $\delta(S) = e[S] / {|S| \choose 2}$. However, it is easy to see that finding a set of vertices $S \subseteq V$ that maximizes $\delta(S)$ is trivial and not interesting, as any pair of vertices connected by an edge forms an optimal solution for this objective function.
For this reason, the formulation known as the \emph{Densest Subgraph Problem} (DSP) aims to find a subset of vertices $S \subseteq V$ that maximizes the \emph{degree density} of $S$, defined as $d(S)= e[S]/|S|$.
Note that this objective function is equivalent to half of the average degree of the vertices within the subgraph (since each edge contributes 2 to the sum of degrees, $e[S]$ is half of the sum of the degrees of the vertices in $S$).
This specific problem has received \rev{a} great deal of attention in the algorithmic literature over the last five decades, with many variants and many applications built on top of this basic definition. To have an idea of the span of research interest
just consider that, while the last couple of years have seen the publication of some groundbreaking results and several interesting contributions, e.g., \cite{harb2022faster,MaCLH22,Luo2023,fang2022densest,boob2020flowless,Chekuri2022supermod,bera2022dynamicsub,liu2022stochastic,Fazzone2022,ma2022convex,veldt2021meandensest,gudapati2021greedy,bonchi2021finding,ma2021directed,ma2021efficient,chekuri2023generalized},
the work by Picard and Queyranne \cite{Picard82}, introducing an exact algorithm for DSP, dates back to 1979 (technical report, although the final paper was officially published in 1982). This survey aims to (1) discuss the fundamental results for DSP in depth, and (2) provide an exhaustive coverage of the many variants proposed in the literature, with a special emphasis on the most recent results.

The survey is organized as follows.
In Section \ref{sec:fundamental} we review the fundamental algorithmic results for DSP, distinguishing between exact algorithms and approximate ones and highlighting connections between DSP and other related algorithmic problems, such as \emph{submodular function minimization} or the computation of the \emph{core decomposition}. Section \ref{sec:cons} covers constrained variants of DSP, including, e.g., constraints on the admissible size of the solution, constraints on set of vertices to be included in the solution, or constraints on the level of connectivity of the solution. Section \ref{sec:variants} instead covers variants that tackle \rev{different but related} objective functions. In Section \ref{sec:graphs} we summarize the massive literature on extracting dense subgraphs from  information-richer graphs such as directed, signed, weighted, probabilistic (or uncertain), multilayer, temporal graphs, etc.
Section \ref{sec:computational} covers \rev{variants of DSP in} different computational settings, such as streaming, distributed, parallel, and MapReduce.
In Section~\ref{sec:apps} we survey real-world applications. Finally, in Section \ref{sec:conclusions} we discuss open problems for future investigation.

%% file: fundamental.tex
In this section, we review the fundamental results for DSP,
starting from classical exact and approximation algorithms
and arriving to very recent breakthrough results.

\subsection{Exact algorithms for DSP}
There is a long history of polynomial-time exact algorithms for DSP.
In 1979, Picard and Queyranne~\cite{Picard82} presented the first exact algorithm for DSP in their technical report (officially published in 1982).
The algorithm is designed based on a series of maximum-flow computations.
In 1984, Goldberg~\cite{goldberg1984finding} gave an improved version of the above algorithm in terms of time complexity.
More than 15 years later, Charikar~\cite{Charikar2000} designed an LP-based exact algorithm for DSP.
As maximum-flow computation can be done by solving LP, the existence of an LP-based exact algorithm is trivial.
However, Charikar's \rev{LP-based algorithm} \rev{exploits a different idea}  \rev{and it is of independent interest, because it allows
to see the renowned greedy-peeling approximation algorithm (described later in Section \ref{subsec:approx})
as a primal-dual algorithm for DSP.}

\subsubsection{Goldberg's maximum-flow-based algorithm}
We next review the maximum-flow-based exact algorithm for DSP, designed by Goldberg~\cite{goldberg1984finding}.
Strictly speaking, the algorithm we describe is slightly different from Goldberg's one, but the difference is not essential and it is just for the sake of simplicity of presentation. The algorithm maintains upper and lower bounds on the (unknown) optimal value of the problem,
and tightens the bounds step-by-step using binary search, until the current lower bound is guaranteed to be the optimal value of DSP.
To update upper and lower bounds, the algorithm utilizes maximum-flow computation.
As initial upper and lower bounds, $m/2$ and $0$ can be employed, respectively.
Let $\beta \geq 0$ be the midpoint of the upper and lower bounds kept in the current iteration.
For $G=(V,E)$ and $\beta \geq 0$, the algorithm constructs the following edge-weighted directed graph $(U,A,w_\beta)$: $U=V\cup \{s,t\}$, $A=A_s\cup A_E\cup A_t$, where
$A_s=\{(s,v)\mid v\in V\}$, $A_E=\{(u,v), (v,u)\mid \{u,v\}\in E\}$, and $A_t=\{(v,t)\mid v\in V\}$, \rev{and} $w_\beta \colon A\rightarrow \mathbb{R}_+$ \rev{such that}
\begin{align*}
w_\beta(e) =
\begin{cases}
\frac{\deg(v)}{2} &\text{if } e=(s,v)\in A_s,\\
\frac{1}{2} &\text{if } e\in A_E,\\
\beta &\text{if } e\in A_t. 
\end{cases}
\end{align*}
For $v\in V$, $\deg(v)$ is the degree of $v$ in $G$. 
This graph is illustrated in Figure~\ref{fig:digraph}. 
Here we introduce some terminology and notation.
An $s$--$t$ cut of $(U,A,w_\beta)$ is a partition $(X,Y)$ of $U$ (i.e., $X\cup Y=U$ and $X\cap Y=\emptyset$) such that $s\in X$ and $t\in Y$.
The cost of an $s$--$t$ cut $(X,Y)$, denoted by $\text{cost}(X,Y)$, is defined as the sum of weights of edges going from $X$ to $Y$,
i.e., $\mathrm{cost}(X,Y)=\sum_{(u,v)\in A: u\in X,v\in Y}w_\beta(u,v)$.
An $s$--$t$ cut having the minimum cost is called a minimum $s$--$t$ cut.
The following lemma is useful for updating the upper and lower bounds using a minimum $s$--$t$ cut of $(U,A,w_\beta)$:
\begin{figure}[t]
\centering
\includegraphics[scale=1.15]{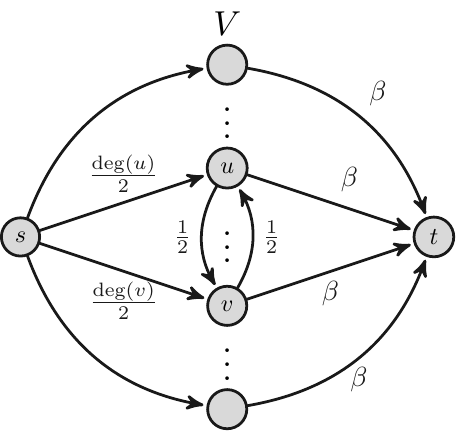}
\caption{An edge-weighted directed graph $(U,A,w_\beta)$ constructed from $G$ and $\beta$.}
\label{fig:digraph}
\end{figure}

\begin{lemma}\label{lemma1}
Let $(X,Y)$ be an $s$--$t$ cut of the edge-weighted directed graph $(U,A,w_\beta)$,
and $S=X\setminus \{s\}$.
Then it holds that
$\mathrm{cost}(X,Y)=m+\beta |S|-e[S].$
In particular, when $X=\{s\}$, $\mathrm{cost}(X,Y)=m$ holds.
\end{lemma}
\begin{proof}
Any edge from $X$ to $Y$ is contained in exactly one of the following sets:
$\{(s,v)\in A_s\mid v\in V\setminus S\},\
\{(u,v)\in A_E\mid u\in S,\, v\in V\setminus S\},\
\{(v,t)\in A_t\mid v\in S\}$.
Therefore, the cost of $(X,Y)$ can be evaluated as follows:
\begin{align*}
\text{cost}(X,Y)
&=\sum_{(s,v)\in A_s:\, v\in V\setminus S}w_\beta(s,v)+\sum_{(u,v)\in A_E:\, u\in S,v\in V\setminus S}w_\beta(u,v)
+\sum_{(v,t)\in A_t:\, v\in S}w_\beta(v,t)\\
&= \sum_{v\in V\setminus S} \frac{\text{deg}(v)}{2}+\frac{|\{{\{u,v\}\in E\mid u\in S,\, v\in V\setminus S}\}|}{2}+\beta|S|\\
&=e[V\setminus S]+|\{{\{u,v\}\in E\mid u\in S,\, v\in V\setminus S}\}|+\beta|S|\\
&=m+\beta|S|-e[S].
\end{align*}
\end{proof}
Let $(X,Y)$ be the minimum $s$--$t$ cut of $(U,A,w_\beta)$ computed by the algorithm.
If $X= \{s\}$ does not hold,
then $S=X\setminus \{s\}$ is a vertex subset that satisfies $\beta|S|-e[S]\leq 0$, i.e., $e[S]/|S|\geq \beta$;
hence, the lower bound on the optimal value can be replaced by $\beta$.
On the other hand, if $X=\{s\}$, then we see that there is no $S\subseteq V$ that satisfies $e[S]/|S|> \beta$;
hence, the upper bound can be replaced by $\beta$.
The following lemma is useful for specifying the termination condition of the algorithm.
\begin{lemma}
Let $G=(V,E)$ be an undirected graph.
For any $S_1,S_2\subseteq V$, if $d(S_1)\neq d(S_2)$, then
$|d(S_1)-d(S_2)|\geq \frac{1}{n(n-1)}$.
\end{lemma}
\begin{proof}
Defining $\Delta=\left| d(S_1)-d(S_2)\right|$, we have
$\Delta=\left| \frac{e[S_1]|S_2|-e[S_2]|S_1|}{|S_1||S_2|}\right|$.
If $|S_1|=|S_2|$ holds, then
$\Delta=\left| \frac{e[S_1]-e[S_2]}{|S_1|}\right|$, implying that $\Delta \geq \frac{1}{n}$.
On the other hand, if $|S_1|\neq |S_2|$, then $|S_1||S_2|\leq n(n-1)$; therefore, $\Delta \geq \frac{1}{n(n-1)}$.
\end{proof}
From this lemma, we see that once the difference between the upper and lower bounds becomes less than $\frac{1}{n(n-1)}$,
the current lower bound attains the optimal value of DSP.
The entire procedure is summarized in Algorithm~\ref{alg:flow}.
The following theorem guarantees the solution quality and time complexity.
\begin{algorithm}[t]
\caption{Maximum-flow-based algorithm}\label{alg:flow}
\SetKwInOut{Input}{Input}
\SetKwInOut{Output}{Output}
\Input{\ $G=(V,E)$}
\Output{\ $S\subseteq V$}
$\beta^{(0)}_\text{ub}\leftarrow m/2$, $\beta^{(0)}_\text{lb}\leftarrow 0$, $i\leftarrow 0$\;
\While{$\beta_\mathrm{ub}^{(i)} - \beta_\mathrm{lb}^{(i)}\geq \frac{1}{n(n-1)}$}{
  $\beta^{(i)}\leftarrow \frac{\beta_\text{lb}^{(i)}+\beta_\text{ub}^{(i)}}{2}$\;
  Compute the minimum $s$--$t$ cut $(X^{(i)},Y^{(i)})$ of $(U,A,w_{\beta^{(i)}})$ using maximum-flow computation\;
  \If{$X^{(i)}\neq \{s\}$}{
    $\beta_\text{lb}^{(i+1)}\leftarrow \beta^{(i)}$, $\beta_\text{ub}^{(i+1)}\leftarrow \beta_\text{ub}^{(i)}$\;
  }
  \Else{
    $\beta_\text{lb}^{(i+1)}\leftarrow \beta_\text{lb}^{(i)}$, $\beta_\text{ub}^{(i+1)}\leftarrow \beta^{(i)}$\;
  }
  $i\leftarrow i+1$\;
}
Compute the minimum $s$--$t$ cut $(X,Y)$ of $(U,A,w_{\beta_\text{lb}^{(i)}})$ using maximum-flow computation\;
\Return{$X\setminus \{s\}$}
\end{algorithm}
\begin{theorem}
Algorithm~\ref{alg:flow} returns an optimal solution to DSP in $O(T_\mathrm{Flow}\log n)$ time, where $T_\mathrm{Flow}$ is the time complexity required to compute a minimum $s$--$t$ cut $(X^{(i)},Y^{(i)})$ using maximum-flow computation ($i=0,1,\dots$).
\end{theorem}
\begin{proof}
From the above discussion, we know that the algorithm returns an optimal solution to DSP.
In what follows, we show that the time complexity of the algorithm is given by $O(T_\text{Flow}\log n)$.
The number of iterations $\hat{i}$ of the while-loop is the minimum (nonnegative) integer among $i$'s that satisfy
\begin{align*}
\frac{1}{2^i}(\beta^{(0)}_\text{ub}-\beta^{(0)}_\text{lb})< \frac{1}{n(n-1)}.
\end{align*}
Therefore, we see that
\begin{align*}
\hat{i}=O(\log(mn(n-1)))=O(\log n).
\end{align*}
Obviously, the time complexity of each iteration of the while-loop is dominated by $T_\text{Flow}$.
Thus, we have the theorem.
\end{proof}

For example, employing the maximum-flow computation algorithm by \rev{Cheriyan et} al.~\cite{Cheriyan96},
we can compute a minimum $s$--$t$ cut of $(U,A,w_{\beta^{(i)}})$ in $O(|U|^3/\log|U|)=O(n^3/\log n)$ time.
Therefore, in that case, the time complexity of Algorithm~\ref{alg:flow} becomes $O(n^3)$.

In the context of maximum-flow-based exact algorithms for DSP, we remark that \rev{Gallo et} al.~\cite{gallo1989fast} showed how a single parametric maximum flow computation provides an optimal solution to DSP.
To the best of our knowledge, the fastest algorithm for parametric maximum flow computation is the one given by Hochbaum~\cite{hochbaum2008pseudoflow} and solves DSP in $\bigO\left(mn \log n\right)$ time, which is better than $O(n^3)$ above, when the input graph is sparse, i.e., $m=O(n)$.

\subsubsection{Charikar's LP-based algorithm}
The algorithm first solves an LP and then constructs a vertex subset, using the information of the optimal solution to the LP,
which is guaranteed to be an optimal solution to DSP.
Let us introduce a variable $x_e$ for each $e\in E$ and a variable $y_v$ for each $v\in V$.
The LP used in the algorithm is as follows:
\begin{alignat*}{4}
&\text{maximize} &\quad &\sum_{e\in E}x_e \\
&\text{subject to} &    &x_e\leq y_u,\ x_e\leq y_v  &\quad &(\forall e=\{u,v\}\in E),\\
&                  &    &\sum_{v\in V}y_v=1, \\
&                  &    &x_e\geq 0, \ y_v\geq 0     &      &(\forall e\in E,\, \forall v\in V).
\end{alignat*}
Roughly speaking, the first constraints stipulate that if we take edge $e=\{u,v\}$, we have to take both of the endpoints $u$ and $v$.
The second constraint just standardizes the objective function of DSP.
The following lemma implies that the LP is a continuous relaxation of DSP.
\begin{lemma}\label{lem:OPT_LP}
Let $\mathrm{OPT}_\mathrm{LP}$ be the optimal value of the LP.
For any $S\subseteq V$, it holds that
$\mathrm{OPT}_\mathrm{LP}\geq d(S).$
In particular, letting $S^*\subseteq V$ be an optimal solution to DSP, we have $\mathrm{OPT}_\mathrm{LP}\geq d(S^*)$.
\end{lemma}
\begin{proof}
Take an arbitrary $S\subseteq V$. Construct a solution $(\bm{x},\bm{y})$ of the LP as follows:
\begin{align*}
x_e=
\begin{cases}
\frac{1}{|S|}  &\text{if } e\in E[S],\\
0              &\text{otherwise},
\end{cases}
\quad
y_v=
\begin{cases}
\frac{1}{|S|}  &\text{if } v\in S,\\
0              &\text{otherwise}. 
\end{cases}
\end{align*}
Then it is easy to see that this solution is feasible for the LP.
Moreover, the objective value of $(\bm{x},\bm{y})$ can be evaluated as
\begin{align*}
\sum_{e\in E}x_e=\sum_{e\in E[S]}\frac{1}{|S|}=\frac{e[S]}{|S|}=d(S).
\end{align*}
Therefore, we have $\text{OPT}_\text{LP}\geq d(S)$.
As we took $S\subseteq V$ arbitrarily, this inequality holds even for $S=S^*$.
\end{proof}

The algorithm first solves the LP to obtain an optimal solution $(\bm{x}^*,\bm{y}^*)$.
For $(\bm{x}^*,\bm{y}^*)$ and $r\geq 0$, let $S(r)=\{v\in V\mid y^*_v\geq r\}$.
Also, let $y^*_\text{max} =\max_{v\in V}y^*_v$.
Then the algorithm computes $r^*\in \argmax\{d(S(r))\mid r\in [0,\,y^*_\text{max}]\}$ and just outputs $S(r^*)$.
To find such $r^*$, it suffices to check the value of $d(S(r))$ at $r=y^*_v$ for every $v\in V$,
because $S(r)$ may change only at those points.
The entire procedure is summarized in Algorithm~\ref{alg:LP}.
Clearly, the algorithm runs in polynomial time.
\begin{algorithm}[t]
\caption{LP-based algorithm}\label{alg:LP}
\SetKwInOut{Input}{Input}
\SetKwInOut{Output}{Output}
\Input{\ $G=(V,E)$}
\Output{\ $S\subseteq V$}
Compute an optimal solution $(\bm{x}^*,\bm{y}^*)$ to the LP\;
$r^*\in \argmax\{d(S(r))\mid r\in [0,\,y^*_\text{max}]\}$\;
\Return{$S(r^*)$}
\end{algorithm}
\begin{theorem}\label{thm:LP}
Algorithm~\ref{alg:LP} is a polynomial-time exact algorithm for DSP.
\end{theorem}
\begin{proof}
Let $S^*\subseteq V$ be an optimal solution to DSP.
By Lemma~\ref{lem:OPT_LP}, we know that $\text{OPT}_\text{LP}\geq d(S^*)$.
Noting the choice of parameter $r^*$ in Algorithm~\ref{alg:LP},
it suffices to show that there exits $r\in [0,\,y^*_\text{max}]$ such that
\begin{align*}
d(S(r))=\frac{e[S(r)]}{|S(r)|}\geq \text{OPT}_\text{LP}.
\end{align*}

Suppose for contradiction that for any $r\in [0,\,y^*_\text{max}]$, it holds that
\begin{align*}
d(S(r))=\frac{e[S(r)]}{|S(r)|}< \text{OPT}_\text{LP}.
\end{align*}
Then we have
\begin{align}\label{ineq:dense_suppose_bar}
\int_0^{y^*_\text{max}}e[S(r)]\,\text{d}r<\text{OPT}_\text{LP}\int_0^{y^*_\text{max}}|S(r)|\,\text{d}r.
\end{align}
Define indicator functions $X_e\colon [0,\,y^*_\text{max}]\rightarrow \{0,1\}$ ($e=\{u,v\}\in E$)
and $Y_v\colon [0,\,y^*_\text{max}]\rightarrow \{0,1\}$ ($v\in V$) as follows:
\begin{align*}
X_e(r)=
\begin{cases}
1  &\text{if } r\leq y^*_u \text{ and } r\leq y^*_v,\\
0  &\text{otherwise},
\end{cases}\quad
Y_v(r)=
\begin{cases}
1  &\text{if } r\leq y^*_v,\\
0  &\text{otherwise}.
\end{cases}
\end{align*}
From the optimality of $(\bm{x}^*,\bm{y}^*)$, we have $x^*_e=\min\{y^*_u,\, y^*_v\}$ for every $e=\{u,v\}\in E$,
and therefore
\begin{align}\label{eq:dense_indicator_x}
\int_0^{y^*_\text{max}}e[S(r)]\,\text{d}r&=\int_0^{y^*_\text{max}}\left(\sum_{e\in E}X_e(r)\right)\text{d}r\nonumber \\
&=\sum_{e\in E}\int_0^{y^*_\text{max}} X_e(r)\,\text{d}r
=\sum_{e\in E}\min\{y^*_u,\,y^*_v\}
=\sum_{e\in E}x^*_e
=\text{OPT}_\text{LP}.
\end{align}
Similarly we have
\begin{align}\label{eq:dense_indicator_y}
\int_0^{y^*_\text{max}}|S(r)|\,\text{d}r&=\int_0^{y^*_\text{max}}\left(\sum_{v\in V}Y_v(r)\right)\text{d}r
=\sum_{v\in V}\int_0^{y^*_\text{max}} Y_v(r)\,\text{d}r
=\sum_{v\in V}y^*_v=1.
\end{align}
By Inequality~\eqref{ineq:dense_suppose_bar} and Equalities~\eqref{eq:dense_indicator_x} and \eqref{eq:dense_indicator_y}, we obtain
\begin{align*}
\text{OPT}_\text{LP}=
\int_0^{y^*_\text{max}}e[S(r)]\,\text{d}r
<\text{OPT}_\text{LP}\int_0^{y^*_\text{max}}|S(r)|\,\text{d}r
=\text{OPT}_\text{LP},
\end{align*}
a contradiction. This concludes the proof.
\end{proof}

It is worth noting that Balalau et al.~\cite{balalau2015topkoverlapping} proved that $\{v\in V\mid y^*_v >0\}$ is a densest subgraph, based on a more sophisticated analysis. This simplifies the above rounding procedure and reduces its time complexity from $O(m+n\log n)$ to $O(n)$, while the time complexity of Algorithm~\ref{alg:LP} is dominated by that for solving the LP.

\subsubsection{An exact algorithm based on submodular function minimization}
The objective function of DSP  has a strong connection with submodularity property of set functions.
Based on this observation, an exact algorithm for DSP can be designed using submodular function minimization. Let $V$ be a finite set.
A function $f\colon 2^V\rightarrow \mathbb{R}$ is said to be submodular
if $f(X)+f(Y)\geq f(X\cup Y)+f(X\cap Y)$ holds for any $X,Y\subseteq V$.
There is a well-known equivalent definition of the submodularity:
a function $f$ is submodular
if and only if $f$ satisfies $f(X\cup \{v\})-f(X)\geq f(Y\cup \{v\})-f(Y)$ for any $X\subseteq Y$ and $v\in V\setminus Y$,
which is called the diminishing marginal return property.
A function $f\colon 2^V\rightarrow \mathbb{R}$ is said to be supermodular if $-f$ is submodular.
A function $f\colon 2^V\rightarrow \mathbb{R}$ is said to be modular if $f$ is submodular and supermodular.
In the submodular function minimization problem (without any constraint),
given a finite set $V$ and a submodular function $f\colon 2^V\rightarrow \mathbb{R}$,
we are asked to find $S\subseteq V$ that minimizes $f(S)$.
Note that the function $f$ is given as a value oracle, which returns $f(S)$ for a given $S\subseteq V$.
The submodular function minimization problem can be solved exactly in polynomial time~\cite{fujishige2005submodular}.
Strictly speaking, there exist algorithms that solve the problem, using the polynomial number of calls of the value oracle in terms of the size of $V$. 
Indeed, Gr\"otschel et al.~\cite{groetschel1981ellipsoid} gave the first polynomial-time algorithm based on the ellipsoid method.
Later, Iwata et al.~\cite{iwata2001combinatorial} \rev{and Schrijver~\cite{Schrijver00}} designed a combinatorial, strongly-polynomial-time algorithm.
Today some faster algorithms are known, e.g., \cite{orlin2009faster,lee2015faster}.

Let us consider the following problem:
$\text{minimize}\ \ \beta |S| - e[S] \, \text{subject to}\ \ S\subseteq V,$
where $\beta \geq 0$ is a constant.
As the function $\beta |S|$ is modular and the function $-e[S]$ is submodular~\cite{fujishige2005submodular}, the function $\beta |S| - e[S]$ is submodular.
Therefore, this problem is a special case of the submodular function minimization problem.
Let us recall the $i$th iteration of the while-loop of Algorithm~\ref{alg:flow},
where either the upper bound $\beta_\text{ub}^{(i)}$ or the lower bound $\beta_\text{lb}^{(i)}$ is updated
using a minimum $s$--$t$ cut of the edge-weighted directed graph $(U,A,w_\beta^{(i)})$,
where $\beta^{(i)}=(\beta_\text{ub}^{(i)}+\beta_\text{ub}^{(i)})/2$.
We can see that the update can be done through solving the above problem.
Let $S_\text{out}$ be an exact solution to the problem with $\beta=\beta^{(i)}$.
If $S_\text{out}\neq \emptyset$ holds,
then $S_\text{out}$ is a vertex subset that satisfies $\beta^{(i)}|S_\text{out}|-e[S_\text{out}]\leq 0$, i.e., $e[S_\text{out}]/|S_\text{out}|\geq \beta^{(i)}$;
hence, the lower bound $\beta_\text{lb}^{(i)}$ can be replaced by $\beta^{(i)}$.
On the other hand, if $S_\text{out}=\emptyset$, then for any $S\subseteq V$,
$e[S]/|S|\leq \beta^{(i)}$ holds; hence, $\beta_\text{ub}^{(i)}$ can be replaced by $\beta^{(i)}$.

\subsection{Approximation algorithms for DSP}\label{subsec:approx}
Despite the polynomial-time solvability of DSP,
there exists a wide literature studying faster approximation algorithms for DSP.
The first approximation algorithm was identified by Kortsarz and Peleg~\cite{Kortsarz-Peleg94}, \rev{based on the concept of $k$-core}.
\rev{For $G=(V,E)$ and $k\in \mathbb{Z}_{+}$, the $k$-core is the (unique) maximal subgraph in which every vertex has degree at least $k$ \cite{malliaros2020core}.
The authors proved that the $k$-core with the maximum $k$ is a $2$-approximate solution for DSP (straightforward after proving Theorem~\ref{thm:peeling} about the greedy peeling algorithm, which is given below).}
Charikar~\cite{Charikar2000} then studied a greedy algorithm for DSP, later to be known as the greedy peeling algorithm,
and proved that the algorithm is also a $2$-approximation algorithm but runs in $O(m+n\log n)$ time.
It is now widely known that the algorithm can be implemented to run in linear time, \rev{i.e., $O(m+n)$ time,} for unweighted graphs.
About 20 years later, Boob et al.~\cite{boob2020flowless} designed an iterative version of the greedy peeling algorithm,
inspired by the multiplicative weights update method~\cite{arora2012multiplicative}.
Later, Chekuri et al.~\cite{Chekuri2022supermod} proved that the iterative greedy peeling algorithm converges to an optimal solution to DSP.
\rev{In the rest of this section, we review these fundamental results.}

\subsubsection{Greedy peeling algorithm}

Here we review the $2$-approximation algorithm for DSP, presented by Charikar~\cite{Charikar2000}.
\rev{The algorithm works by repeatedly removing from the graph the vertex with the smallest degree, and saving the remaining vertices as a potential solution. The process continues until only one node remains in the graph. Finally, it outputs the subset of vertices with the highest density among the recorded candidate solutions.}
The pseudocode is given in Algorithm~\ref{alg:peeling}, where for $S\subseteq V$ and $v\in S$, $\deg_S(v)$ denotes the degree of $v$ in $G[S]$.
\begin{algorithm}[t]
\caption{Greedy peeling algorithm}\label{alg:peeling}
\SetKwInOut{Input}{Input}
\SetKwInOut{Output}{Output}
\Input{\ $G=(V,E)$}
\Output{\ $S\subseteq V$}
$S_{n}\leftarrow V$, $i\leftarrow n$\;
\While{$i>1$}{
  $v_\text{min}\in \argmin\{\text{deg}_{S_i}(v)\mid v\in S_i\}$\;
  $S_{i-1}\leftarrow S_i\setminus \{v_\text{min}\}$\;
  $i\leftarrow i-1$\;
}
$S_\mathrm{max} \in \argmax\{d(S)\mid S\in \{S_1,\dots, S_{n}\}\}$\;
\Return{$S_\mathrm{max}$}
\end{algorithm}

\begin{theorem}\label{thm:peeling}
Algorithm~\ref{alg:peeling} is a $2$-approximation algorithm for DSP. Moreover, the algorithm can be implemented to run in linear time.
\end{theorem}
\begin{proof}
Let $S^*\subseteq V$ be an optimal solution to DSP.
From the optimality of $S^*$, for any $v\in S^*$, we have
\begin{align*}
d(S^*)=\frac{e[S^*]}{|S^*|}\geq \frac{e[S^*\setminus \{v\}]}{|S^*|-1}=d(S^*\setminus \{v\}).
\end{align*}
Transforming the above inequality using $e[S^*\setminus \{v\}]=e[S^*]-\deg_{S^*}(v)$, we have that for any $v\in S^*$,
\begin{align}\label{ineq:dense_optimality}
\deg_{S^*}(v)\geq d(S^*).
\end{align}

Let $v^*$ be the vertex that is contained in $S^*$ and removed first by Algorithm~\ref{alg:peeling}.
Let $S'\subseteq V$ be the vertex subset kept just before removing $v^*$ in Algorithm~\ref{alg:peeling}.
Then the density of $S'$ can be evaluated as
\begin{align}\label{ineq:peeling_LB}
d(S')=\frac{\frac{1}{2}\sum_{v\in S'}\deg_{S'}(v)}{|S'|}
\geq \frac{\frac{1}{2}|S'|\deg_{S'}(v^*)}{|S'|}
\geq \frac{1}{2}\deg_{S^*}(v^*)
\geq \frac{1}{2}d(S^*),
\end{align}
where the first inequality follows from the greedy choice of $v^*\in S'$,
the second inequality follows from $S'\supseteq S^*$,
and the last inequality follows from Inequality~\eqref{ineq:dense_optimality}.
Noticing that $S'$ is one of the candidate subsets of the output, we have the approximation ratio of $2$, as desired.

Next we show that the algorithm can be implemented to run in linear time.
For each (possible) degree $d=0,\dots, n-1$, the vertices having degree $d$ are kept in a doubly-linked list.
In the very first iteration, the algorithm scans the lists with increasing order of $d$ until it finds a vertex.
Once a vertex is found, the algorithm removes the vertex from the list, and moves each neighbor of the vertex to the one lower list.
Owing to the structure of doubly-linked lists, this entire operation can be done in $O(\deg(v))$ time.
Moving to the next iteration, it suffices to go back to the one lower list because no vertex can exist in lower levels than that,
implying that in the entire algorithm, the number of moves from one list to another is bounded by $O(n)$.
Therefore, the above operations can be conducted in $O(m+n)$ time.
Note that we do not need to compute the density of a vertex subset from scratch in each iteration.
It suffices to compute $d(V)$ in the very first iteration and then compute the density of a current vertex subset based on the difference:
in each iteration, the numerator decreases by the degree of the removed vertex in the graph at hand and the denominator decreases by 1.
This concludes the proof.
\end{proof}

Here we demonstrate that the $k$-core with the maximum $k$ is also a $2$-approximate solution for DSP. 
This can be shown easily using Inequality~\eqref{ineq:dense_optimality}.
Let $k^*$ be the maximum value among $k$'s such that there exists a nonempty $k$-core. 
Let $S^*_\text{core}$ be (the vertex subset of) the $k^*$-core.
By the definition of $k^*$-core, we have $\min_{v\in S^*_\text{core}}\deg_{S^*_\text{core}}(v)=k^* \geq \deg_{S^*}(v^*)$,
where $v^*$ is that defined in the proof of Theorem~\ref{thm:peeling}.
Then, using Inequality~\eqref{ineq:dense_optimality}, we have
\begin{align*}
d(S^*_\text{core})
\geq \frac{1}{2}k^*
\geq \frac{1}{2}\deg_{S^*}(v^*)
\geq \frac{1}{2}d(S^*),
\end{align*}
as desired.

In practice, Algorithm~\ref{alg:peeling} rarely outputs a vertex subset that has an objective value of almost half the optimal value,
which means that the algorithm has a better empirical approximation ratio.
However, Gudapati et al.~\cite{gudapati2021greedy} proved that the approximation ratio of $2$ is tight,
that is, there exists no constant $\alpha < 2$ such that the algorithm is an $\alpha$-approximation algorithm.

It should be remarked that Algorithm~\ref{alg:peeling} can be seen as a primal-dual algorithm for DSP, which also gives a proof of $2$-approximation.
Let us consider the dual of the LP (used in Algorithm~\ref{alg:LP}):
\begin{alignat*}{4}
&\text{minimize}   &\quad &t \\
&\text{subject to} &    &t\geq \sum_{e=\{u,v\}\in E} z_{e,v}      &\quad      &(\forall v\in V),\\
&                  &    &z_{e,u} + z_{e,v} \geq 1 &      &(\forall e=\{u,v\}\in E),\\
&                  &    &z_{e,u}, z_{e,v}\geq 0    &      &(\forall e=\{u,v\}\in E).
\end{alignat*}
The above dual LP can be understood as follows:
each edge $e=\{u,v\}$ has cost $1$ and we need to distribute it to the endpoints $u,v$ so as to minimize the maximum cost over vertices.
Recall that in each iteration of Algorithm~\ref{alg:peeling}, we specify the minimum degree vertex $v_\text{min}$ in the graph at hand and then remove it.
If we assign all the costs of the incident edges of $v_\text{min}$ to $v_\text{min}$ when removing it, we can get a feasible solution for the dual LP.
Let $t^*$ be the objective value of the solution (in terms of the dual LP).
Let us focus on the iteration where the assignment of the maximum cost $t^*$ is carried out.
Then the density of the vertex subset kept just before the assignment is lower bounded by $t^*/2$, because all vertices have degree at least $t^*$.
By the LP's (weak) duality theorem, $t^*$ is lower bounded by the optimal value of the primal LP, thus by the optimal value of DSP.
Therefore, we again see that Algorithm~\ref{alg:peeling} is a $2$-approximation algorithm for DSP.

\subsubsection{Iterative greedy peeling algorithm and beyond}\label{subsubsec:iterative_peeling}

Boob et al.~\cite{boob2020flowless} proposed {\sc Greedy++} for DSP,
an algorithm inspired by the multiplicative weights update~\cite{arora2012multiplicative}.
Let $T\in \mathbb{Z}_{+}$.
The algorithm runs Charikar's greedy peeling algorithm for $T$ times
while updating the priority of each vertex using the information of the past iterations.
The pseudocode is given in \refalg{peeling++}.
Note that each iteration of {\sc Greedy++} requires $\bigO(m+n\log n)$ time. 
The authors empirically showed that the algorithm tends to converge to the optimum
and conjectured its convergence to optimality.
\begin{algorithm}[t]
\caption{\textsf{Greedy++}}\label{alg:peeling++}
\SetKwInOut{Input}{Input}
\SetKwInOut{Output}{Output}
\Input{\ $G=(V,E)$, $T\in \mathbb{Z}_{>0}$}
\Output{\ $S\subseteq V$}
\lForEach{$v \in V$} {$\ell_v \leftarrow 0$}
\For{$t =1,2,\dots,T$}{
  $S^t_{n}\leftarrow V$, $i\leftarrow n$\;
  \While{$i> 1$}{
    $v_\text{min}\in \argmin\{\text{deg}_{S^t_i}(v) + \ell_v \mid v\in S^t_i\}$\;
    $\ell_{v_\text{min}} \leftarrow \ell_{v_\text{min}} + \text{deg}_{S^t_i}(v_\text{min})$\;
    $S^t_{i-1}\leftarrow S^t_i\setminus \{v_\text{min}\}$\;
    $i\leftarrow i-1$\;
  }
  $t\leftarrow t+1$\;
}
$S_\mathrm{max} \in \argmax\{d(S)\mid S\in \{S^1_n,\dots, S^1_1,S^2_n,\dots, S^2_1,\dots, S^T_n,\dots, S^T_1\}\}$\;
\Return{$S_\mathrm{max}$}
\end{algorithm}

Chekuri et al.~\cite{Chekuri2022supermod} proved the conjecture of Boob et al.~\cite{boob2020flowless},
showing that {\sc Greedy++} converges to a solution with an approximation ratio arbitrarily close to $1$
and that it naturally extends to a broad class of supermodular functions (i.e., normalized, non-negative, and monotone supermodular functions) in the numerator of the objective function.
More in detail, the authors proved that for any $\epsilon >0$,
{\sc Greedy++} provides a
$(1+\epsilon)$-approximate solution for DSP
after
$\bigO\left(\frac{\Delta \log n}{\text{OPT} \epsilon^2} \right)$ iterations,
where $\Delta$ is the maximum degree of a vertex in the graph and $\text{OPT}$ is the optimal value of DSP. 
\rev{Harb et al.~\cite{harb2023convergence} recently established the convergence of {\sc Greedy++} to the so-called optimal dense decomposition vector.}
Fazzone et al.~\cite{Fazzone2022} modified {\sc Greedy++} to have a quantitative certificate of the solution quality provided by the algorithm at each iteration.
Thanks to this, the authors equipped {\sc Greedy++} with a practical device
that allows termination whenever a solution with a user-specified approximation ratio is found, making the algorithm suited for practical purposes.
Prior to Chekuri et al.~\cite{Chekuri2022supermod},
Boob et al.~\cite{boob2019faster} provided a $(1+\epsilon)$-approximation algorithm for DSP
by reducing DSP to solving $\bigO\left(\log n\right)$ instances of the mixed packing and covering problem.
This algorithm runs in $\tilde\bigO\left(\frac{m \Delta}{\epsilon} \right)$ time\footnote{$\poly\log n$ factors are hidden in the $\tilde\bigO$ notation.}.
Chekuri et al.~\cite{Chekuri2022supermod} designed another $(1+\epsilon)$-approximation algorithm running in
$\tilde\bigO\left(\frac{m}{\epsilon} \right)$ time by approximating maximum flow.
Very recently, Harb et al.~\cite{harb2022faster} proposed a new iterative algorithm,
which provides an $\epsilon$-additive approximate solution for DSP in
$\bigO\left(\frac{\sqrt{m \Delta}}{\epsilon} \right)$ iterations,
where each iteration requires $\bigO\left(m\right)$ time and admits some level of parallelization.
The authors also provided a different peeling technique called fractional peeling, with theoretical guarantees and good empirical performance.

%% file: cons.tex
This section delves into variants of DSP, defined by means of constraints.
We will begin covering size-constrained problems, which set limits on the desired size of the output \rev{subset}.
Next, we will investigate seed-set problems, where an initial set of nodes --- referred to as seed set --- guides the search for the densest subgraph.
Finally, we will examine problems with connectivity constraints, where the output subgraph must meet specific requirements in terms of connectivity, to prevent potential vertex/edge failures.

\subsection{Size constraints}
\label{sec:size}
Size-constrained versions of DSP are \rev{well studied} in literature, \rev{as they 
find applications in several real-world contexts}. 
While DSP is polynomial-time \rev{solvable}, adding size constraints makes it \NP-hard.

\subsubsection{Densest $k$-subgraph problem}\label{subsubsec:DkS}
Given a simple undirected graph $G=(V,E)$ and a positive integer $k$, the \emph{Densest $k$-Subgraph problem} (D$k$S) requires to find a vertex subset $S\subseteq V$ that maximizes $d(S)=e[S]/|S|$ subject to $|S|=k$.
As the size of solutions is fixed, the objective function can be reduced to $e[S]$.
It is easy to see that the maximum clique problem can be reduced to D$k$S; therefore, the problem is \NP-hard.
D$k$S is known not only as a variant of DSP but also as one of the most fundamental combinatorial optimization problems, and would deserve a survey of its own. 

Feige et al.~\cite{FPK01} proposed a combinatorial polynomial-time $O(n^{1/3-\delta})$-approximation algorithm for some tiny $\delta >0$.
Later, Goldstein and Langberg~\cite{Goldstein2009dense} estimated the above approximation ratio of $O(n^{1/3-\delta})$ and concluded that it is approximately equal to $O(n^{0.3226})$.
In addition, they presented an algorithm with a slightly better approximation ratio of $O(n^{0.3159})$.
Bhaskara et al.~\cite{Bhaskara+10} proposed an $O(n^{1/4+\epsilon})$-approximation algorithm running in $n^{O(1/\epsilon)}$ time, for any $\epsilon >0$.
This approximation ratio is the current \rev{state-of-the-art} for D$k$S.
The algorithm is based on a clever procedure that distinguishes random graphs from \rev{some} random graphs with planted dense subgraphs.

In addition, there are some algorithms with an approximation ratio depending on the parameter $k$.
Asahiro et al.~\cite{Asahiro2000} demonstrated that the straightforward application of the greedy peeling algorithm attains the approximation ratio of $O(n/k)$.
Later, Feige and Langberg~\cite{FL01} employed semidefinite programming (SDP) and achieved an approximation ratio somewhat better than $O(n/k)$.

There exist also approximation algorithms for specific instances of D$k$S.
Arora et al.~\cite{AKK95} proved that there exists a polynomial-time approximation scheme (PTAS) for D$k$S on $G$ with $m=\Omega(n^{2})$ and $k=\Omega(n)$.
Ye and Zhang~\cite{Ye2003approximation} developed an SDP-based polynomial-time $1.7048$-approximation algorithm for D$k$S with $k=n/2$,
which improves some previous results, e.g., a trivial randomized $4$-approximation or $2$-approximation algorithm using LP~\cite{Goemans1996mathematical} or SDP~\cite{Feige1997densest}.
Liazi et al.~\cite{Liazi2008constant} presented a polynomial-time $3$-approximation algorithm for D$k$S with chordal graphs.
Later, Chen et al.~\cite{Chen2010densest} developed a polynomial-time constant-factor approximation algorithm for D$k$S with a variety of classes of intersection graphs, including chordal graphs and claw-free graphs.
They also proposed a PTAS for D$k$S on unit disk graphs, which improves the previous $1.5$-approximation for D$k$S on proper interval graphs, a special case of unit disk graphs~\cite{Backer2010constant}.
Papailiopoulos et al.~\cite{Papailiopoulos2014finding} designed an algorithm for D$k$S that looks into a low-dimensional space of dense subgraphs. The approximation guarantee depends on the graph spectrum, which is effective for many graphs in applications. The algorithm runs in nearly-linear time under some mild assumptions of the graph spectrum, and is highly parallelizable.
Khanna and Louis~\cite{khanna2020planted} introduced semi-random models of instances with a planted dense subgraphs and studied the approximability of D$k$S.
They showed that approximation ratios better than $O(n^{1/4+\epsilon})$ can be achieved for a wide range of parameters of the models.

The literature is also rich in inapproximability results for D$k$S.
It is known that D$k$S has no PTAS under some reasonable computational complexity assumptions.
For example, Feige~\cite{Feige02} assumed that random 3-SAT formulas are hard to refute,
Khot~\cite{Khot06} assumed that \NP does not have any randomized algorithm running in subexponential time,
and Raghavendra and Steurer~\cite{raghavendra2010expansion} assumed a strengthened version of the unique games conjecture (UGC).
\rev{Bhaskara et al.~\cite{Bhaskara+12} studied the inapproximability of D$k$S from the perspective of SDP relaxations and devised lower bounds on the integrality gaps of strong SDP relaxations, showing that beating a factor of $n^{\Omega(1)}$ is a barrier even for the most powerful SDPs.}
Manurangsi~\cite{Manurangsi2017Almost} proved that D$k$S cannot be approximated up to a factor of $n^{\frac{1}{\left(\log \log n\right)^c}}$, 
for some $c>0$ assuming the Exponential Time Hypothesis (ETH)~\cite{Impagliazzo2001Complexity}.
Braverman et al.~\cite{braverman2017eth} ruled out a PTAS in terms of the additive approximation, assuming ETH.
\rev{Very recently, Chuzhoy et al.~\cite{Chuzhoy+23} showed that D$k$S cannot be approximated up to a factor of $2^{\log^\epsilon n}$ for some $\epsilon >0$, assuming a novel conjecture on the hardness of some constraint satisfaction problems.}

Besides the approximability and inapproximability, the parameterized complexity of D$k$S has also been studied
(see e.g., \cite{Cygan2015parameterized} for the foundations of the parameterized complexity).
Cai~\cite{Cai2008parameterized} proved that D$k$S is $\text{W}[1]$-hard with respect to the parameter $k$, meaning that there exists no fixed-parameter tractable algorithm for D$k$S parameterized by $k$, unless $\ensuremath{\mathrm{P}} = \NP$, \rev{which was later strengthened by Komusiewicz and Sorge~\cite{komusiewicz2015algorithmic}}. 
Bourgeois et al.~\cite{Bourgeois2013exact} showed that D$k$S can be solved exactly in $2^\texttt{tw}\cdot n^{O(1)}$ time, where $\texttt{tw}$ is the tree-width of the input graph.
Broersma et al.~\cite{Broersma2013tight} demonstrated that D$k$S can be solved in $k^{O(\texttt{cw})}\cdot n$ time, where $\texttt{cw}$ is the clique-width of the input graph, but it cannot be solved in $2^{o(\texttt{cw}\log k)}\cdot n^{O(1)}$ time, unless ETH fails.
\rev{Komusiewicz and Sorge~\cite{komusiewicz2015algorithmic} developed an algorithm for D$k$S with a time complexity of
$O((4.2(\Delta-1))^{k-1} (\Delta+k)\, k^{2} n)$, where $\Delta$ is the maximum degree of a vertex in the graph.
The algorithm is randomized and can only produce false negatives with probability at most $1/e$.}
Recently, Mizutani and Sullivan~\cite{Mizutani2022parameterized} proved that D$k$S can be solved in any of $f(\texttt{nd})\cdot n^{O(1)}$ time and $O(2^\texttt{cd}\cdot k^2n)$ time, where \texttt{nd} and \texttt{cd} denote, \rev{respectively, two characteristics of the input graph, i.e., the neighborhood diversity~\cite{Lampis12} and the cluster deletion number~\cite{Bocker+13},  while} $f$ is some computable function.
Hanaka~\cite{Hanaka2023computing} independently showed that D$k$S is fixed parameter tractable \rev{using} the neighborhood diversity, \rev{as well as other parameters}.

There are some other exact algorithms for D$k$S,
based on mathematical programming, heuristic search, or graph-theoretic methods.
Billionnet et al.~\cite{Billionnet2009improving} devised a reformulation technique that is applicable to a wide range of quadratic programming problems, including D$k$S.
Malick and Roupin~\cite{Malick2012solving} presented a branch-and-bound method based on SDP relaxations of D$k$S.
Later, Krislock et al.~\cite{Krislock2016computational} improved the above branch-and-bound method using a better bounding procedure.
Komusiewicz and Sommer~\cite{komusiewicz2020fixcon} devised an enumeration-based exact algorithm for a special case of D$k$S, where the subgraph obtained should be connected. 
Gonzales and Migler~\cite{gonzales2019densest} designed an $O(nk^2)$-time exact algorithm for D$k$S on outerplanar graphs.

Kawase and Miyauchi~\cite{Kawase-Miyauchi18} introduced the concept of dense-frontier points of a graph.
Given $G=(V,E)$, plot all the points contained in $\{(|S|,e[S])\mid S\subseteq V\}$.
They referred to the extreme points of the upper convex hull of the above set as the dense frontier points of $G$.
Note that the densest subgraph and the entire graph are typical vertex subsets corresponding to dense frontier points.
\rev{The authors} designed an LP-based algorithm that computes a corresponding vertex subset for every dense frontier point, in polynomial time.
An algorithm designed by Nagano et al.~\cite{Nagano+11} can also be used \rev{for the same purpose}.

Finally, several effective heuristics for D$k$S have been proposed.
Sotirov~\cite{sotirov2020solving} developed coordinate descent heuristics for D$k$S. 
Bombina and Ames~\cite{bombina2020convex} introduced a novel convex programming relaxation for D$k$S using the nuclear norm relaxation of a low-rank and sparse decomposition of the adjacency matrices of subgraphs with $k$ vertices. Using the relaxation, they proved that an optimal solution can be obtained if the input is sampled randomly from a distribution of random graphs constructed to contain a highly dense subgraphs with $k$ vertices with high probability.
Konar and Sidiropoulos~\cite{konar2021exploring} reformulated D$k$S as a submodular function minimization subject to a cardinality constraint, and introduced a relaxation as a Lov\'asz extension minimization over the convex hull of the cardinality constraint. They proposed an effective heuristic for D$k$S by developing a highly scalable algorithm for the relaxation, based on the \rev{Alternating Direction Method of Multipliers} (ADMM). 
\rev{Arrazola et al.~\cite{Arrazola+18} demonstrated that Gaussian boson sampling, one of the limited models of quantum computation, can be used for enhancing randomized algorithms for D$k$S.}

There are some papers dealing with the edge-weighted version of D$k$S, where the edge weights are nonnegative.
Ravi et al.~\cite{Ravi1994heuristic} presented a polynomial-time $4$-approximation algorithm for the problem with edge weights satisfying the triangle inequality.
Later, Hassin et al.~\cite{Hassin1997approximation} proposed an algorithm with a better approximation ratio of $2$.
Recently, Chang et al.~\cite{chang2020hardness} considered a generalized setting, where edge weights only satisfy a relaxed variant of the triangle inequality.
They demonstrated that the problem is \NP-hard for any degree of relaxation of the triangle inequality and extended the above $2$-approximation algorithm to the generalized setting.
Barman~\cite{barman2018approximating} studied a slight variant of D$k$S called the Normalized Densest $k$-Subgraph problem (ND$k$S), where the objective function is normalized to be at most 1, i.e., $e[S]/|S|^2$ is considered. As $|S|$ is fixed to $k$, the inapproximability in terms of the multiplicative approximation is inherited from D$k$S.
Instead, Barman~\cite{barman2018approximating} focused on an additive approximation for ND$k$S, and developed an $\epsilon$-additive approximation algorithm running in $n^{O(\log \Delta / \epsilon^2)}$ time, where $\Delta$ is the maximum degree of $G$.
Barman~\cite{barman2018approximating} also gave an $\epsilon$-additive approximation algorithm for a bipartite graph variant of ND$k$S called the \rev{Densest $k$-Bipartite Subgraph} problem (D$k$BS).
Braverman et al.~\cite{braverman2017eth} studied the inapproximability in terms of the additive approximation of ND$k$S:
\rev{they} ruled out an additive PTAS for ND$k$S under ETH.
Hazan and Krauthgamer~\cite{Hazan2011how} showed that there exits no PTAS for D$k$BS under some computational complexity assumptions.

\subsubsection{Densest at-least-$k$-subgraph problem and densest at-most-$k$-subgraph problem}
There are two relaxations of D$k$S, introduced by Andersen and Chellapilla~\cite{AndersenChellapilla}.
The two problems are called the \rev{Densest at-least-$k$ Subgraph} problem (Dal$k$S) and the \rev{Densest at-most-$k$ Subgraph} problem (Dam$k$S).
As suggested by the names, Dal$k$S and Dam$k$S ask $S\subseteq V$ that maximizes $d(S)$ subject to $|S|\geq k$ and $|S|\leq k$, respectively.
Obviously, similar to D$k$S, Dam$k$S is \NP-hard.
The \NP-hardness of Dal$k$S is not trivial, but it was proved by Khuller and Saha~\cite{Khuller2009Dense} by reducing D$k$S to Dal$k$S.

Andersen and Chellapilla~\cite{AndersenChellapilla} presented a linear-time $3$-approximation algorithm for Dal$k$S
using the greedy peeling algorithm for DSP.
The analysis of the approximation ratio of $3$ is based on the relationship between subgraphs with large density and subgraphs with large minimum degree,
which can be viewed as a generalization of the analysis of $2$-approximation for DSP by Kortsarz and Peleg~\cite{Kortsarz-Peleg94}.
Andersen and Chellapilla~\cite{AndersenChellapilla} also mentioned the hardness of approximation of Dam$k$S:
in particular, they proved that if there exists a polynomial-time $\gamma$-approximation algorithm for Dam$k$S,
then there exists a polynomial-time $8\gamma^2$-approximation algorithm for D$k$S.

Later, Khuller and Saha~\cite{Khuller2009Dense} improved the approximability of Dal$k$S and the hardness of approximation of Dam$k$S.
For Dal$k$S, they designed an LP-based $1/2$-approximation algorithm, 
\rev{which solves a series of $n-k+1$ LPs, constructs $n-k+1$ candidate solutions from the optimal solutions to the LPs, and outputs the best among them.} 
Their analysis of the approximation ratio of $2$ is based on the analysis by Charikar~\cite{Charikar2000}
that DSP can be solved exactly by the LP-based algorithm.
To \rev{decrease the number of LPs to solve from $n-k+1$ to $1$, the authors} suggested incorporating the algorithm by Andersen and Chellapilla into (a lighter version of) the above algorithm.
As for the inapproximability of Dam$k$S, Khuller and Saha~\cite{Khuller2009Dense} proved that
if there exists a polynomial-time $\gamma$-approximation algorithm for Dam$k$S,
then there exists a polynomial-time $4\gamma$-approximation algorithm for D$k$S,
which improves the above result by Andersen and Chellapilla~\cite{AndersenChellapilla} and implies that Dam$k$S is as hard to approximate as D$k$S within a constant factor.
Recently, Zhang and Liu~\cite{zhang2021approximating} developed a randomized bicriteria approximation algorithm for Dam$k$S.

\subsection{Seed set}\label{subsec:seed}
Dai et al.~\cite{dai2022anchored} studied the problem of Anchored Densest Subgraph search (ADS). Given a graph $G=(V,E)$, a reference node set $R$ and an anchored node set $A$, with $A \subseteq R \subseteq V$, the ADS problem consists in finding a set of nodes $S$ that contains all nodes in $A$ \rev{but} maximizes the following quantity:
$(2 e[S] - \sum_{v \in S \setminus R} \deg(v)) / |S|$.
The authors designed an exact algorithm for ADS based on a modified version of the Goldberg~\cite{goldberg1984finding} reduction to
$O\left(\log \sum_{v \in R} \deg(v) \right)$
instances of maximum-flow.
The complexity of the proposed method is bounded by a polynomial of $\sum_{v \in R} \deg(v)$, and is independent of the size of the input graph.

Sozio and Gionis~\cite{Sozio} defined and studied the \emph{cocktail party} problem:
given a graph $G=(V,E)$,
and a set of query nodes $Q \subseteq V$,
the problem consists in finding
a set of nodes $S$ containing all query nodes ($Q\subseteq S$),
 whose induced subgraph
maximizes
a node-monotone non-increasing function
satisfying
a set of monotone non-increasing properties.
The authors considered the minimum degree as the node-monotone non-increasing function to maximize
and a condition on the maximum allowed distance between the solution set $S$ and the query set $Q$ as the monotone non-increasing property to satisfy.
The author proved that a direct adaptation to their problem of the greedy peeling algorithm (Algorithm~\ref{alg:peeling}) always returns an optimal solution.

Fazzone et al.~\cite{Fazzone2022} defined the Dense Subgraphs with Attractors and Repulsers problem (DSAR), whose goal is to find
a dense cluster of nodes $S$, where each node in $S$ is simultaneously close to
a given set $A$ of nodes (Attractors) and far from a given set $R$ of nodes (Repulsers). They proved that DSAR is a special instance of a generalization of DSP on weighted graphs (see Section \ref{subsec:labeled}).

\subsection{Connectivity constraints}
Connectivity constraints on DSP require the output subgraph to satisfy specific requirements in terms of connectivity between the vertices.
These problems are motivated by the fact that densest subgraphs have a structural drawback, that is, they may not be robust to vertex/edge failure, thus not reliable in real-world applications such as network design, transportation, and telecommunications, to name a few.
Indeed, densest subgraphs may not be well-connected, which implies that they may be disconnected by removing only a few vertices/edges within it.
As a toy example, consider a barbell graph consisting of two equally-sized large cliques bridged by only one edge.
In the classical DSP setup, the entire graph would be the densest subgraph, but the failure of the edge connecting the two cliques would disconnect the entire graph.

In this spirit, Bonchi et al.~\cite{bonchi2021finding} introduced two related optimization problems: the densest $k$-vertex-connected subgraph problem and the densest $k$-edge-connected subgraph problem.
In the densest $k$-vertex/edge-connected subgraph problem,
given an undirected graph $G=(V,E)$ and a positive integer $k$, we seek a vertex subset $S\subseteq V$ that maximizes $d(S)$ subject to the constraint that $G[S]$ is $k$-vertex/edge-connected, i.e., the subgraph would be still connected with the removal of any subset of $k$ different vertices/edges.
Bonchi et al.~\cite{bonchi2021finding} first designed polynomial-time $\left(4/\gamma,1/\gamma\right)$-bicriteria approximation algorithms
with parameter $\gamma\in [1,2]$ for these problems.
Note that setting $\gamma=1$, we can obtain $4$-approximation algorithms for the problems.
The algorithms are designed based on a well-known theorem in extremal graph theory, proved by Mader~\cite{Mader72}.
They then designed polynomial-time $19/6$-approximation algorithms for the problems,
which improves the above approximation ratio of $4$ derived directly from the bicriteria approximation ratio.

%% file: variants.tex
The objective function of \rev{DSP}, i.e., the degree density,
has been generalized to various forms for extracting a more sophisticated structure in a graph.
Section~\ref{subsec:numerator} covers variants that generalize the numerator $e[S]$ of the density,
while Section~\ref{subsec:denominator} discusses variants that generalize the denominator $|S|$.
Section~\ref{subsec:generalization_others} reviews generalizations that do not fall in the above categorization.

\subsection{Generalizing the numerator}\label{subsec:numerator}
Tsourakakis~\cite{Tsourakakis15} generalized the notion of density to the $k$-clique density.
For $G=(V,E)$ and $S\subseteq V$, the $k$-clique density for some fixed positive integer $k$ is defined as
$h_k(S)=c_k(S)/|S|$,
where $c_k(S)$ is the number of $k$-cliques contained in $G[S]$.
Obviously, when $k=2$, it reduces to the original density.
In the $k$-clique densest subgraph problem ($k$-clique DSP),
given an undirected graph $G=(V,E)$, we are asked to find $S\subseteq V$ that maximizes $h_k(S)$.
In particular, when $k=3$, the problem is referred to as the triangle densest subgraph problem (triangle DSP).
Tsourakakis~\cite{Tsourakakis15} proved that unlike many optimization problems for detecting a large near-clique,
the $k$-clique DSP is polynomial-time solvable when $k$ is constant.
Indeed, the author designed a maximum-flow-based exact algorithm and a supermodular-function-maximization-based exact algorithm
for the problem with constant $k$.
The author also demonstrated that a generalization of the greedy peeling algorithm,
which in each iteration removes a vertex participating in the minimum number of $k$-cliques, attains $k$-approximation.
Computational experiments show that \rev{even the triangle densest subgraphs are much closer to large near-cliques, compared with the densest subgraphs}. 

Later Mitzenmacher et al.~\cite{mitzenmacher2015scalable} conducted a follow-up work.
Their work is motivated by the fact that it is prohibitive to compute an exact or even well-approximate solution
to the $k$-clique DSP for reasonably large $k$ (e.g., $k>3$) on large graphs, due to the expensive cost of counting $k$-cliques.
To overcome this issue, they presented a sampling scheme called the densest subgraph sparsifier,
yielding a randomized algorithm that produces a well-approximate solution to the $k$-clique DSP
while providing significantly reduced time and space complexities.
Specifically, the sampling scheme samples each $k$-clique independently with an appropriate probability,
which can be incorporated as a preprocessing in any algorithm for the problem.
In addition to the sampling scheme, they also devised two simpler exact algorithms for the $k$-clique DSP.
Finally, the authors extended the $k$-clique DSP to the bipartite graph setting.
For an undirected bipartite graph $G=(L\cup R, E)$, positive integers $p,q$, and $S\subseteq L\cup R$, they defined the $(p,q)$-biclique density as $b_{p,q}(S)=c_{p,q}(S)/|S|$,
where $c_{p,q}(S)$ is the number of $(p,q)$-cliques contained in $G[S]$.
In the $(p,q)$-biclique densest subgraph problem ($(p,q)$-biclique DSP), given an undirected bipartite graph $G=(L\cup R, E)$,
we seek $S\subseteq L\cup R$ that maximizes $b_{p,q}(S)$.
They showed that all the above results for the $k$-clique \rev{DSP} can be extended to the $(p,q)$-biclique DSP.
\rev{Numerical results} demonstrate that the proposed sampling-based algorithms output near-optimal solutions to the problems. 
\rev{Related to the $(p,q)$-biclique DSP, Sar\i{}y\"{u}ce and Pinar~\cite{sariyuce2018peeling} devised algorithms that unfold dense bipartite subgraphs with hierarchy of relations.}

Fang et al.~\cite{Fang2019Efficient} devised more efficient exact and approximation algorithms for the $k$-clique DSP.
To this end, they introduced a generalization of the $k$-core called the $(k,\Psi)$-core.
For a positive integer $k$ and an $h$-clique $\Psi$, a $(k,\Psi)$-core is a maximal subgraph
in which every vertex participates in at least $k$ $h$-cliques.
Therefore, if we take a 2-clique (i.e., an edge) as $\Psi$, the concept reduces to the ordinary $k$-core.
Note that the concept of $(k,\Psi)$-core is a special case of $k$-$(r,s)$ nucleus
introduced by Sar\i{}y\"{u}ce et al.~\cite{Sariyuce2015nucleus,Sariyuce2017nucleus}.
Using the concept, their exact algorithm for the $k$-clique DSP runs as follows:
It computes lower and upper bounds on the $k$-clique density value for each $(\ell,\Psi)$-core computed,
and based on those bounds, it derives lower and upper bounds on the optimal value of the problem.
Then the algorithm specifies some $(\ell,\Psi)$-cores that may contain an optimal solution to the problem,
and solve the problem on them.
A useful fact here is that such $(\ell,\Psi)$-cores tend to be much smaller than the original graph,
enabling us to compute an optimal solution in much shorter time in practice.
On the other hand, \rev{the design of} their approximation algorithm is based on the fact
that the $(\ell,\Psi)$-core with the maximum value of $\ell$ is a good approximation to an optimal solution.
\rev{
Sun et al.~\cite{sun2020kclist++} also developed efficient exact and approximation algorithms for the $k$-clique DSP.
One of their approximation algorithms, called \textsc{kClist++}, is an iterative algorithm, which processes just one $k$-clique in a graph in each iteration.
The algorithm does not deal with all $k$-cliques at a time, and just requires the space complexity linear in $n$ and $m$.
They analyzed the convergence rate of the algorithm.
Another approximation algorithm, called \textsc{Seq-Sampling++}, is inspired by the above algorithm by Mitzenmacher et al.~\cite{mitzenmacher2015scalable} and based on sampling of $k$-cliques. 
}

Recently, Gao et al.~\cite{gao2022colorful} designed a graph reduction technique to accelerate approximation algorithms for the $k$-clique DSP.
To this end, they introduced the novel concept called the colorful $h$-star.
Assume that the vertices of a graph are colored so that any pair of vertices having an edge receives different colors.
Then, for a positive integer $h$, a colorful $h$-star is a star contained in the graph as a (not necessarily induced) subgraph
in which all vertices have different colors.
Note that the colorful $h$-star is a relaxed concept of $h$-clique; indeed, every $h$-clique is a colorful $h$-star.
They showed that unlike $k$-cliques, the colorful $h$-stars can be counted efficiently using a newly devised dynamic programming method, and designed an efficient colorful $h$-star core decomposition algorithm.
Based on this, they designed a graph reduction technique to accelerate any approximation algorithm for the $k$-clique DSP.
Moreover, they showed that the colorful $h$-star core itself can be a good heuristic solution for the $k$-clique DSP.

Konar and Sidiropoulos~\cite{konar2022triangle} studied the Triangle Densest $k$-Subgraph problem (TD$k$S).
The problem is a variant of D$k$S, where given an undirected graph $G=(V,E)$ and a positive integer $k$,
we are asked to find $S\subseteq V$ that maximizes the $3$-clique density $h_3(S)=c_3(S)/|S|$ (or simply $c_3(S)$) subject to $|S|=k$.
They showed that TD$k$S is \NP-hard \rev{and presented} a heuristic algorithm based on a mirror descent algorithm for a convex relaxation derived by the Lov\'asz extension. 
The proposed algorithm is shown to be empirically \rev{effective, 
and moreover}, it sometimes obtains a better solution even in terms of D$k$S than state-of-the-art algorithms for D$k$S.

Bonchi et al. \cite{BonchiKS19} generalized the notion of density by considering the $h$-degree of a node, i.e., the number of other nodes at distance no more than $h$ from the node. Based on this they defined the distance-$h$ densest subgraph problem and showed that, analogously to the $h = 1$ case, the inner-most core of the core decomposition, i.e., the $(k,h)$-core such that there is no non-empty $(j,h)$-core with $j>k$, provides an approximation to the distance-$h$ densest subgraph.

Miyauchi and Kakimura~\cite{Miyauchi-Kakimura18} aims to find a community,
i.e., a dense subgraph that is only sparsely connected to the rest of the graph, based on DSP.
They generalized the density as
$d_\alpha(S)=\frac{e[S]-\alpha\cdot e[S,\overline{S}]}{|S|}$ ($\alpha \in [0,\infty)$), 
where $\alpha$ is a nonnegative parameter and $e[S,\overline{S}]$ is the cut size of $S$, i.e., the number of edges between $S$ and $V\setminus S$.
\rev{This} quality function penalizes the connection between $S$ and $V\setminus S$, resulting in a preferential treatment for community structure.
The authors studied the problem of maximizing this quality function,
and designed \rev{LP-based and maximum-flow-based exact algorithms}.
Moreover, they presented a \rev{greedy peeling algorithm} with some quality guarantee.
Computational experiments demonstrate that the proposed algorithms are highly effective in finding community structure in a graph.

Recently, Chekuri et al.~\cite{Chekuri2022supermod} introduced the \rev{Densest Supermodular Subset} problem (DSS),
where given a finite set $V$ and a nonnegative \rev{normalized} supermodular function $f:2^V\rightarrow \mathbb{R}_+$,
we are asked to find $S\subseteq V$ that maximizes $f(S)/|S|$.
As $e[S]$ is a supermodular function over $V$ \rev{for a} given $G=(V,E)$, this problem is a generalization of DSP.
For \rev{DSS}, they presented a natural generalization of the iterative greedy peeling algorithm for DSP (reviewed in Section~\ref{subsubsec:iterative_peeling}).
Their significant contribution is the proof of the fact that the generalized algorithm outputs a $(1+\epsilon)$-approximate solution for DSS 
after $O\left(\frac{\Delta_f \log n}{\lambda^* \epsilon^2}\right)$ iterations,
where $\Delta_f=\max_{v\in V}(f(V)-f(V\setminus \{v\}))$ and $\lambda^*$ is the optimal value of the problem.

\rev{
Very recently, Huang et al.~\cite{Huang+23} introduced a generalization of DSS,
called the densest supermodular subset with possible negative values,
where the supermodular function is not guaranteed to be nonnegative.
Consequently, the model contains some existing problems that are not caught by DSS, 
e.g., the above problem by Miyauchi and Kakimura~\cite{Miyauchi-Kakimura18}.
For the model, the authors gave a simple reduction to DSS in terms of exact computation, 
meaning that it can be solved exactly through solving DSS.
Then they designed the first strongly-polynomial exact algorithm for DSS (thus for the generalization), based on Dinkelbach's algorithm~\cite{Dinkelbach67}.
Specifically, unlike the existing weakly-polynomial exact algorithms based on binary search over the objective values, the algorithm actively specifies the sequence of supermodular maximization problems to be solved.
}

\subsection{Generalizing the denominator}\label{subsec:denominator}
Kawase and Miyauchi~\cite{Kawase-Miyauchi18} addressed the size issue of DSP, 
\rev{which means that when solving} DSP,
it may happen that the obtained subset is too large or too small in comparison with the size desired in the application at hand.
As mentioned in Section~\ref{sec:cons}, \rev{DSP has} size-constrained variants, e.g., D$k$S, Dal$k$S, and Dam$k$S,
which explicitly specify the size range.
Unlike these variants, \rev{the authors} generalized the density without putting any constraint.
Specifically, they introduced the $f$-density of $S\subseteq V$, which defined as $e[S]/f(|S|)$,
where $f\colon \mathbb{Z}_+\rightarrow \mathbb{R}_+$ is a monotonically non-decreasing function.
Note that earlier than this, Yanagisawa and Hara~\cite{yanagisawa2018discounted} introduced an intermediate generalization
called the discounted average degree, i.e., $e[S]/|S|^\alpha$ for $\alpha \in [1,2]$.
In the $f$-\rev{Densest Subgraph} problem ($f$-DS), we are asked to find $S\subseteq V$ that maximizes the $f$-density $e[S]/f(|S|)$.
Although the $f$-DS does not explicitly specify the size of vertex subsets, the above size issue can be handled using convex or concave function $f$ appropriately.
Indeed, the authors showed that any optimal solution to $f$-DS with convex/concave function $f$ has a size smaller/larger than or equal to that of a densest subgraph.
For the $f$-DS with convex $f$, they proved the \NP-hardness with some concrete $f$,
and designed a polynomial-time
$\min\left\{\frac{f(2)/2}{f(S^*)/|S^*|^2},\, \frac{2f(n)/n}{f(|S|^*)-f(|S^*|-1)}\right\}$-approximation algorithm,
where $S^*\subseteq V$ is an optimal solution to the $f$-DS.
The approximation ratio \rev{is} complicated but it reduces to a simpler form by considering some concrete $f$, e.g., $f(x)=x^\alpha$ ($\alpha \in [1,2]$) and $f(x)=\lambda x + (1-\lambda)x^2$ ($\lambda \in [0,1]$).
For the $f$-DS with concave $f$, they designed \rev{LP-based and maximum-flow-based exact algorithms}.
In particular, the LP-based exact algorithm computes not only an optimal solution but also vertex subsets corresponding to dense frontier points, as explained in Section~\ref{sec:cons}.
Finally, \rev{they showed that the greedy peeling algorithm achieves $3$-approximation}.

\subsection{Other variants}\label{subsec:generalization_others}
Tsourakakis et al.~\cite{tsourakakis2013denser} defined the Optimal Quasi-Clique (OQC) problem of finding \rev{$S\subseteq V$} that maximizes $e[S] - \alpha \binom{|S|}{2}$, 
thus trying to find a subgraph \rev{that} is denser in terms of the edge density $e[S]/{|S|\choose 2}$, instead of the degree density adopted by DSP.

Recently, Veldt et al.~\cite{veldt2021meandensest} generalized the density to the single-parameter family of quality functions.
Specifically, they introduced the $p$-density for $S\subseteq V$, based on the concept of generalized mean (also called power mean or $p$-mean) of real values, as 
$M_p(S)=\left(\frac{1}{|S|}\sum_{v\in S}\deg_S(v)^p\right)^{1/p}$ ($p\in [-\infty,\infty]$),
where for $p\in \{-\infty, 0, \infty\}$, $M_p(S)$ is defined as its limit,
i.e., $M_{-\infty}(S)=\lim_{p\rightarrow -\infty}M_p(S)=\min_{v\in S}\deg_S(v)$, $M_{0}(S)=\lim_{p\rightarrow 0}M_p(S)=\prod_{v\in S}\deg_S(v)$,
and $M_{\infty}(S)=\lim_{p\rightarrow \infty}M_p(S)=\max_{v\in S}\deg_S(v)$.
When $p=1$, the $p$-density reduces to the original density.
The generalized mean densest subgraph problem asks for finding $S\subseteq V$ that maximizes $M_p(S)$.
It is worth mentioning that the generalized mean densest subgraph problem deals with DSP and the problem of finding $k$-core with maximum $k$ in a unified manner ($p=1$ and $p=-\infty$, respectively).
\rev{The authors} first proved that when $p\geq 1$, the \rev{problem} can be solved exactly in polynomial time by repeatedly solving submodular function minimization.
They then designed a faster $(p+1)^{1/p}$-approximation algorithm based on the greedy peeling algorithm.
They specified a class of graphs for which this approximation ratio is tight, and showed that as $p\rightarrow \infty$, the approximation ratio converges to $1$.
They also proved that for any $p >1$, the greedy peeling algorithm for DSP outputs an arbitrarily bad solution to the problem, on some graph classes.

\rev{Very recently, Chekuri and Torres~\cite{chekuri2023generalized} proved that the generalized mean densest subgraph problem is NP-hard for any $p \in \left(-\frac{1}{8}, 0\right) \cup \left(0, \frac{1}{4}\right)$ and even for any $p \in \left(-3,0\right) \cup \left(0,1\right)$ if there are positive edge weights.
On the other hand, the authors demonstrated that any optimal solution to DSP or an output of some modified version of the greedy peeling algorithm gives a tight $2$-approximation for the problem with $p < 1$.
The authors noted that for $p \geq 1$, solving the generalized mean densest subgraph problem is equivalent to solving DSS (in Section~\ref{subsec:numerator}) with supermodular function 
$f_p(S)=\sum_{v\in S}\deg_S(v)^p$.
As a result, the iterative greedy peeling algorithm for DSS, specialized for the generalized mean densest subgraph problem,
converges to an optimal solution, 
requiring $O(mn)$ time per iteration.  
To reduce this running time, the authors
designed a faster implementation of the greedy peeling algorithm designed by Veldt et al.~\cite{veldt2021meandensest} approximating its solution by a factor of
$(1 + \epsilon)$
while reducing the running time to
$O\left(\frac{pm\log^2 n}{\epsilon}\right)$.
}

Balalau et al. \cite{balalau2015topkoverlapping} defined the $(k,\alpha)$-Dense Subgraph with Limited Overlap problem ($(k,\alpha)$-DSLO): given an integer $k >
0$ as well as a real number $\alpha \in [0, 1]$, find at most $k$
subgraphs that maximize the total aggregate density,
i.e., the sum of the average degree of each subgraph,
under the constraint that the maximum pairwise
Jaccard coefficient between the set of nodes in the subgraphs be at most $\alpha$. They proved that the problem is \NP-hard even when $\alpha = 0$ (disjoint subgraphs) and showed that 
\rev{the most intuitive, simple heuristic can produce an arbitrarily bad solution}. 
The authors thus presented an efficient algorithm for $(k,\alpha)$-DSLO which comes with provable guarantees in some cases of
interest.

%% file: graphs.tex
In this section, we survey the literature on \rev{DSP} for different types of graphs. We discuss the extension of the problem to labeled graphs, negatively weighted graphs, directed graphs, multilayer graphs, temporal graphs, uncertain graphs, \rev{hypergraphs, and general metric spaces}.
\input{sec6/labeled}
\input{sec6/negative}

\input{sec6/directed}

\input{sec6/multilayer}

\input{sec6/temporal}

\input{sec6/uncertain}

\input{sec6/hypergraphs}

\input{sec6/metric}

%% file: sec6/labeled.tex
\subsection{Edge- and vertex-labeled graphs}\label{subsec:labeled}
An edge-labeled (\rev{vertex-labeled, resp.}) graph is a type of graph where each edge (\rev{vertex, resp.}) is assigned one or more specific labels or attributes.  
The simplest example of an edge- and/or vertex-labeled graph is a graph where each edge and/or each vertex is associated with a scalar attribute, i.e., a weight. \rev{DSP} on weighted graphs was already studied by Goldberg~\cite{goldberg1984finding}, 
but recently dusted-off by Fazzone et al.~\cite{Fazzone2022}, that dubbed it as Heavy and Dense Subgraph Problem (HDSP).
Given a simple undirected graph $(G, V, E, w_V, w_E)$, where $w_V: V \rightarrow \mathbb{R}_+$ and $w_E: E \rightarrow \mathbb{R}_+$,
HDSP requires to find \rev{$S \subseteq V$} that maximizes
$(e[S] + \sum_{s \in S} w_V(s)) /|S|$, where $e[S] = \sum_{e \in E[S]} w_E(e)$.
HDSP is \rev{polynomial-time solvable}: 
Goldberg~\cite{goldberg1984finding} gave an exact polynomial-time algorithm based on a reduction to the $s$-$t$ maximum flow / $s$-$t$ minimum cut problem.
Recently, Fazzone et al.~\cite{Fazzone2022} studied the approximation guarantees of the greedy peeling algorithm (\rev{Section~\ref{subsec:approx}}) when applied to HDSP, and then adapted the iterative greedy peeling of Chekuri et al.~\cite{Chekuri2022supermod} to HDSP.

Anagnostopoulos et al.~\cite{anagnostopoulos2020fair} tackled the algorithmic fairness issue \rev{of DSP} by defining the Fair Densest Subgraph Problem (FDSP) on a graph with binary labeling of vertices. 
FDSP \rev{computes} 
$S \subseteq V$
of maximum density while ensuring that
$S$
contains an equal number of representatives of either label, guaranteeing that the binary-protected attribute is not disparately impacted. 
The authors proved that FDSP is \NP-hard and approximating \rev{FDSP} with a polynomial-time algorithm is at least as hard as Dam$k$S, for which no constant-factor approximation algorithms are known (see Section~\ref{sec:size}).
On the other hand, any $\alpha$-approximation to Dam$k$S is a $2\alpha$-approximation to FDSP. 
 In the case where the input graph \rev{is} fair, the authors \rev{introduced} a polynomial-time $2$-approximation algorithm for the problem and \rev{proved} the tightness of this approximation factor under the small set expansion hypothesis~\cite{raghavendra2010expansion}.
The authors also \rev{devised} an algorithm, based on a spectral embedding, \rev{enabling} to retrieve an approximate solution with theoretical guarantees in polynomial time for the case where the input graph is an expander that contains an almost-regular dense subgraph.

\rev{
Very recently, Miyauchi et al.~\cite{Miyauchi+23} studied the problem of finding a densest diverse subgraph in a graph 
whose vertices have different attribute values/types that they refer to as colors. 
Let $C$ be a set of colors. 
The input graph $G=(V,E)$ is associated with a color function $\ell:V\rightarrow C$ that assigns a color to each vertex. 
They proposed two novel optimization models motivated by different realistic scenarios. 
The first model, called the \emph{Densest Diverse Subgraph Problem} (DDSP), maximizes the density 
while guaranteeing that no color represents more than a given fraction $\alpha\in [1/|C|,1]$ of vertices in the output, which is a generalization of the above FDSP. 
Varying the fraction $\alpha$ enables to range the diversity constraint and to interpolate from a diverse dense subgraph where all colors have to be equally represented 
to an unconstrained densest subgraph. 
They designed a scalable $O(\sqrt{n})$-approximation algorithm for the case where $V$ is a feasible solution. 
The second model, called the \emph{Densest at-least-$\vv{k}$-Subgraph problem} (Dal$\vv{k}$S), is motivated by the setting where any specified color should not be overlooked. 
As its name suggests, Dal$\vv{k}$S is a generalization of Dal$k$S, where $\vv{k}$ represents cardinality demands with one coordinate per color class. 
They designed a $3$-approximation algorithm using LP together with an acceleration technique. 
Computational experiments demonstrate that densest subgraphs in real-world graphs have strong homophily in terms of colors, 
emphasizing the importance of considering the diversity, and the proposed algorithms contribute to extracting dense but diverse subgraphs. 
}

Tsourakakis et al.~\cite{tsourakakis2019novel} considered \rev{solving} DSP in edge-(multi)labeled networks with exclusion queries: given a \rev{graph} $G = (V,E)$ in which any edge \rev{is} associated \rev{with} a subset of labels in the label universe $L$, and an input set $l \subseteq L$, find the densest subgraph that contains only edges whose label is contained in $l$. To solve this problem, they resort to an heuristic based on the Densest Subgraph with Negative Weights problem (see \refsec{negative_weights}).

Rozenshtein et al.~\cite{rozenshtein2020mining} studied the problem of finding a dense edge-induced subgraph in an edge-labeled graph whose edges are similar to each other based on a given similarity function of the labels. The authors modeled the problem \rev{using the} objective function that is the sum of \rev{density and similarity terms}, and \rev{introduced} a Lagrangian relaxation problem, for which they \rev{proposed} an algorithm based on parametric min-cut, requiring $\bigO(m^3 \log m)$ time and $\bigO(m^2)$ space.

%% file: sec6/negative.tex
\subsection{Negatively weighted graphs}
\label{sec:negative_weights}
\rev{HDSP, reviewed in the previous section}, was defined on positive weights.
Recently some specific applications (e.g., networks derived from correlation matrices) are requiring to extract dense subgraphs from networks whose edges can be negatively weighted. When we allow negative weights, most of the algorithmic results for DSP no longer \rev{hold}. Cadena et al. \cite{cadena2016dense} were the first to analyze \rev{dense structures} in graphs with negative weights on edges. Building on top of \cite{tsourakakis2013denser} they defined the Generalized Optimal Quasi Clique (GOQC) problem as follows: given a \rev{graph} $G=(V,E)$, a weight function $w: E \rightarrow \mathbb{R}$ and a penalty function $\bar{\alpha}: E \rightarrow \mathbb{R}$, the goal is to find a subset of nodes $S$ that maximizes $f_{\bar{\alpha}}(S) = \sum_{\{u,v\} \in E[S]}(w(u,v) - \bar{\alpha}(u,v))$.
The authors showed that the problem is \NP-hard to approximate within a factor of $\bigO(n^{1/2-\epsilon})$, and proposed an algorithm composed by solving SDP \rev{by refining} the local search algorithm by \cite{tsourakakis2013denser}.

Tsourakakis et al. \cite{tsourakakis2019novel} tackled the Densest Subgraph with Negative Weights \rev{problem (DSNW)} defined as follows: given a graph $G=(V,E)$, a weight function $w: E \rightarrow \mathbb{R}$, the goal is to find a subset of nodes $S$ that maximizes $\sum_{\{u,v\} \in E[S]}w(u,v)/|S|$.
The authors proved that the problem is \NP-hard via a reduction based on the fact that the max-cut problem on graphs with possibly negative edges is \NP-hard.
In order to efficiently solve this problem, they analyzed the performance of the greedy peeling \rev{algorithm} for DSP, proving that for this problem it achieves an approximation of $\frac{\rho^*}{2}-\frac{\Delta}{2}$, where $\rho^*$ is the optimal value and $\Delta$ is the \rev{largest absolute value of the negative degrees}.

%% file: sec6/directed.tex
\subsection{Directed graphs}\label{subsec:directed}

Directed graphs are graphs in which for any edge it is also taken into account the directionality, meaning that any edge is defined by a source vertex $s$ and a target vertex $t$. Therefore, for two vertices $u$ and $v$, there can exist edges $(u,v)$ and $(v,u)$. In this context, the classical notion of density \rev{becomes} meaningless, since it does not take into account the directionality. The first notion of density for directed graphs was proposed by Kannan and Vinay \cite{Kannan}:
given a directed \rev{graph} $G=(V,E)$, the Directed Densest Subgraph problem (DDS) requires to find two set of nodes $S$ and $T$ that maximize $f(S,T) = |E(S,T)| / \sqrt[]{|S|\cdot|T|}$, where $E(S,T)$ is the set of edges having its source in $S$ and target in $T$. \rev{This density notion prioritizes pairs of sets of nodes} with a massive number of connections from one to the other; the square root at the denominator ensures that the output is not a single edge, which would take maximum without it.
\rev{The authors} proposed an $\bigO(\log n)$-approximation algorithm for DDS, by relating the objective function to the singular value of the adjacency matrix and applying Monte Carlo algorithm for its computation. Charikar~\cite{Charikar2000} obtained an exact \rev{algorithm} and a 2-approximation algorithm based on $\bigO(n^2)$ LPs. The latter resembles the greedy peeling for DSP running in $\bigO(n+m)$ \rev{time}, yielding \rev{the total time complexity} of $\bigO(n^2(n+m))$. For both algorithms it was \rev{also} observed that with $\bigO(\log n /\epsilon)$ execution of LPs, it is possible to output $(1+\epsilon)$- and  $(2+\epsilon)$-approximations, \rev{respectively}.
Khuller and Saha \cite{Khuller2009Dense} gave the first max-flow-based polynomial-time algorithm for DDS; likewise for DSP, this algorithm requires $\bigO(\log n)$ executions of max flow instances, or $\bigO(1)$ executions of parametrized maximum flow problem. \rev{The authors also} claimed to have obtained a 2-approximation algorithm with time complexity $\bigO(n+m)$, that was recently disproved by Ma et al. \cite{ma2021directed} by a counter-example, \rev{where} the authors reported an alternative 2-approximation algorithm provided by Ma~\cite{saha2029fix}, whose time complexity is $\bigO(n(n+m))$.
Bahmani et al. \cite{bahmani} developed a $2(1 + \epsilon)$-approximation algorithm, that ends in $\bigO(\log_{1+\epsilon}n)$ passes over the input graph.

Sawlani and Wang \cite{sawlani2020dynamic} reduced DDS to HDSP, relying on the knowledge of the ratio between sizes of $|S|$ and $|T|$ for the optimal solution. Since this quantity is unknown, they \rev{proved} that $O(\log n/\epsilon)$ \rev{solutions} of a $(1 + \epsilon/2)$-approximation algorithm for HDSP \rev{can guarantee} a $(1+\epsilon)$-approximate solution for DDS. Furthermore, the authors introduced different proposals for densest subgraph algorithms in \rev{a} fully dynamic setting for undirected graphs and vertex-weighted undirected graphs; therefore, their reduction of DDS to HDSP led to the first algorithm for DDS in \rev{a} fully dynamic setting, that maintains a $(1+\epsilon)$-approximate solution with worst-case time \rev{$O(\poly(\log n, \epsilon^{-1}))$} per update. Chekuri et al. \cite{Chekuri2022supermod} followed this reduction to adapt their DSP flow-based algorithm for DDS, that outputs a $(1+\epsilon)$-approximation in $\tilde \bigO(m /\epsilon^2)$ \rev{time}.

\rev{Ma et al. \cite{ma2021directed,ma2020efficient} obtained different results in this context.} Introducing the notion of $[x,y]$-core, the directed counterpart of the $k$-core notion for undirected graphs, they \rev{proved} that \rev{an optimal solution to DDS can be identified} through it with theoretical guarantees.
As a first direct consequence, these results enable the maximum-flow based exact algorithm to be executed only on a reduced version of the graph, composed by some $[x,y]$-cores, making it faster. Furthermore, they provided a divide-and-conquer strategy to carefully select the optimal value of $|S|/|T|$, that reduces the possible values from $n^2$ to $k$, with $k \ll n^2$.
Furthermore, they proved that a particular instance of $[x,y]$-core is a 2-\rev{approximate solution}, \rev{and thus proposed} an algorithm to find it that runs in $\bigO(\sqrt{m}(n+m))$.
Ma et al. \cite{ma2022convex} designed a new LP formulation for DDS, with a Frank-Wolfe based algorithm to optimize the associate dual. With this new LP, they were able to design a new algorithmic framework to reduce the number of LP instances to solve. More precisely, they provided a $(1+\epsilon)$-\rev{approximate solution} with \rev{$\bigO(T_\mathrm{FW}\log_{1+\epsilon}n)$}, where $T_\mathrm{FW}$ is the time complexity for solving a single instance of their LP.

%% file: sec6/multilayer.tex
\subsection{Multilayer networks}\label{subsec:multilayer}

Multilayer networks are a generalization of the ordinary (i.e., single-layer) graphs.
For positive integer $\ell$, let $[\ell]=\{1,2,\dots, \ell\}$.
Mathematically, a multilayer network is defined as a tuple $(V,(E_i)_{i\in [\ell]})$,
where $V$ is the set of vertices and each $E_i$ ($i=1,2,\dots, \ell$) is a set of edges on $V$.
That is, a multilayer network has a number of edge sets (called layers),
which may encode different types of connections and/or time-dependent connections over the same set of vertices.
As the density value of $S\subseteq V$ varies layer by layer,
there would be several ways to define the objective function of multilayer-network counterparts of DSP.
For $S\subseteq V$ and $i\in [\ell]$, we denote by $d_i(S)$ the density of $S$ in terms of the layer $i$. Jethava and Beerenwinkel~\cite{jethava2015relational} introduced the first optimization problem for dense subgraph discovery in multilayer networks,
which they referred to as the Densest Common Subgraph Problem (DCSP).
In the problem, given a multilayer network $G=(V,(E_i)_{i\in [\ell]})$, we seek a vertex subset $S\subseteq V$ that maximizes
the minimum density over layers, i.e., $\min_{i\in [\ell]} d_i(S)$.
They devised an LP-based polynomial-time heuristic and a $2\ell$-approximation algorithm based on the greedy peeling.
Reinthal et al.~\cite{reinthal2016finding} studied \rev{a} simplex method and \rev{an} interior-point method for the above LP 
and observed that employing the interior-point method can shorten the computation time in practice.
Galimberti et al.~\cite{galimberti2017core,galimberti2020core} introduced a generalization of DCSP,
which they refer to as the multilayer densest subgraph problem.
This problem aims at optimizing a trade-off between the minimum density value over layers and the number of layers having such a density value.

Later, Charikar et al.~\cite{charikar2018common} designed two combinatorial polynomial-time algorithms with approximation ratios
$O(\sqrt{n\log \ell})$ and $O(n^{2/3})$ (irrespective of $\ell$), respectively.
Moreover, they showed some strong inapproximability results for the problem, based on some assumptions.
Specifically, they showed that \rev{DCSP} is at least as hard to approximate as \textsc{MinRep}, a well-studied minimization version of \textsc{Label Cover}, which implies that the problem cannot be approximated to within a factor of $2^{\log^{1-\epsilon}n}$, unless $\NP \subseteq \text{DTIME}(n^{\textsf{polylog}(n)})$.
They also showed that if the planted dense subgraph conjecture is true, the problem cannot be approximated to within a factor of $n^{1/4-\epsilon}$ and even for $\ell=2$, the problem cannot be approximated to within $n^{1/8-\epsilon}$.

Recently, Hashemi et al.~\cite{hashemi2022firmcore} designed a sophisticated core decomposition algorithm for multilayer networks,
which they call the FirmCore decomposition algorithm.
For $k\in \mathbb{Z}_+$ and $\lambda\in [\ell]$, a subgraph $H=(S,(E_i[S])_{i\in [\ell]})$ is called a $(k,\lambda)$-FirmCore
if it is a maximal subgraph in which every vertex has degree no less than $k$ in the subgraph for at least $\lambda$ layers.
They devised a polynomial-time algorithm for finding the set of $(k,\lambda)$-FirmCores for all possible $k$ and $\lambda$.
They proved that the FirmCore decomposition unfolds an approximate solution to the multilayer densest subgraph problem,
with a better approximation ratio than that obtained by~\cite{galimberti2017core} for many instances.

Very recently, Kawase et al.~\cite{kawase2023stochastic} studied stochastic solutions to dense subgraph discovery in multilayer networks.
Their novel optimization problem asks to find a stochastic solution, i.e., a probability distribution over the family of vertex subsets, rather than a single vertex subset,
whereas it can also be used for obtaining a single vertex subset.
The quality of stochastic solutions is measured using the expectation of the following three metrics, the density, the robust ratio, and the regret,
on the layer selected by the adversary.
Therefore, their optimization problem can be seen as (a generalization of) the stochastic version of DCSP.
Unlike DCSP, their optimization problem can be solved exactly in polynomial time;
indeed, they designed an LP-based polynomial-time exact algorithm.
They proved that the output of the proposed algorithm has a useful structure;
the family of vertex subsets with positive probabilities has a hierarchical structure.
This leads to several practical benefits, e.g., the largest size subset contains all the other subsets and the optimal solution obtained by the algorithm has support size at most $n$.

Finally, we take a look at a very special case of multilayer networks called dual networks, i.e., the case of $\ell=2$ in multilayer networks.
Wu et al.~\cite{wu2015dual,Wu+16} introduced an optimization problem of detecting a dense and connected subgraph in dual networks:
given a dual network $G=(V,(E_1,E_2))$, we are asked to find $S\subseteq V$ that maximizes $d_1(S)$, i.e., the density on the first layer,
under the constraint that $(S,E_2[S])$ is connected.
They proved that the problem is \NP-hard and designed a scalable heuristic.
Later, Chen et al.~\cite{Chen+22} considered a variant of the problem,
where $k\in \mathbb{Z}_+$ is given as an additional input, and we seek $S\subseteq V$ that maximizes the minimum degree of vertices on the first layer,
under the constraint that $(S,E_2[S])$ is $k$-edge-connected, 
enabling to control the strength of connectivity on the second layer.
Owing to the use of the minimum degree, this problem can be solved exactly in polynomial time, unlike the above problem by Wu et al.~\cite{wu2015dual,Wu+16}.
\rev{Feng et al.~\cite{feng2021specgreedy} considred another variant, where we seek a vertex subset that 
maximizes a function prioritizing the density on the first layer while discarding some subsets that are dense even on the second layer.
}

Yang et al. \cite{yang2018mining} proposed the density contrast subgraph problem: Given \rev{a dual network, find a vertex subset that maximizes the difference between densities on the two layers}. They \rev{reduced} the maximum clique problem to their problem, proving that the problem is \NP-hard and cannot be approximated within $O(n^{1-\epsilon})$ for any $\epsilon > 0$. They \rev{solved} the problem via a variant of the greedy peeling, providing an $O(n)$-\rev{approximation}.
Lanciano et al. \cite{lanciano2020contrast} proposed a variant of this problem, by considering maximizing the difference in terms of the number of edges, subject to an input penalty term that controls the output of the solution. They showed that this problem can be mapped to an instance of GOQC (\rev{see} Section \ref{sec:negative_weights}).

%% file: sec6/temporal.tex
\subsection{Temporal networks}
Temporal networks are a special case of multilayer networks, where the different layers \rev{correspond} to different time slices, so that the order of the layers is important. \rev{Another definition of} a temporal network is a graph, where each edge is \rev{associated with} a timestamp from a discrete temporal domain. 
More formally we can define a temporal network as a graph $G = (V, E, T)$ where $V$ denotes the set of vertices, $T = \{0,1,..., t_\text{max}\} \subset \mathbb{N}$ is the time domain, and $E \subseteq V \times V \times T$.

Bodganov et al. \cite{bogdanov2011heavy} were the first to address \rev{a variant of DSP} in temporal networks. Allowing the presence of negatively weighted edges, they defined the problem of finding the heaviest temporal subgraph, i.e., a set of nodes and a time-interval maximizing the edge's weights summation, and \rev{showed} that it is \NP-hard even with edge weights restricted in $\{-1, +1\}$.
The same problem has been recently tackled again by Ma et al. \cite{ma2020temporal}, that proposed a new heuristic \rev{algorithm}. 

Rozenshtein et al. \cite{rozenshtein2017dynamic} proposed to find the densest temporal subgraph in the latter representation of the temporal network, 
by imposing a constraint on the maximum number of timestamp to be included in the time-interval, and another on the maximum length of the span of the time-interval.
Lately, Rozenshtein et al. \cite{Rozenshtein2019segmentation} introduced the $k$-densest-episodes problem, where an episode is defined as a pair $S = \{I, H\}$ with $I = \{i_1, i_2\}$ representing the time-interval and $H \subseteq V$ as a set of nodes. The objective is to find the $k$ episodes that are densest in $k$ disjoint time-intervals.
The same problem has been recently tackled by Dondi and Hosseinzadeh \cite{Dondi2021}, that proposed an heuristic \rev{running} in $O(|T| + (k + \tau) t_\text{DSP})$, where $\tau$ is the maximum number of iterations and $t_\text{DSP}$ is \rev{the time complexity required for solving DSP}.

Angel et al.~\cite{Angel2013} tackled real-time story identification on Twitter via maintenance of densest subgraph in the fully dynamic setting.
Similarly, Bonchi et al. \cite{BonchiBGS16,BonchiBGS19} adopted anomalously dense subgraphs in temporal networks as a way to identify buzzing stories in social media.
To recognize a story as buzzing, it needs to have high density in the interactions (or co-occurrences) among 
all objects (terms or entities) therein and for all time instants in a temporal window.
Thus they defined the density of a subgraph in a given time interval as the minimum degree, among all vertices of the subgraph and all timestamps of the temporal window.
They showed that the problem of finding the densest subgraph, for a given time interval, can be solved exactly by the peeling algorithm to construct the core decomposition, and by returning the innermost core. 
Chu et al. \cite{chu2019online} defined the density bursting subgraph as a subgraph that accumulates its density at the fastest speed in a temporal network, according to their definition of the burstiness of a subgraph.

Semertzidis et al.~\cite{semertzidis2019finding} introduced another generalization of the densest common subgraph problem, called the Best Friends Forever (BFF) problem,
in the context of evolving graphs with a number of snapshots.
The BFF problem is a series of optimization problems that maximize an \emph{aggregate density} over snapshots,
where the aggregate density is set to be the average/minimum value of the average/minimum degree of vertices over layers.
Similarly to the multilayer densest subgraph problem,
they also considered the variant called the On--Off BFF ($\text{O}^2$BFF) problem, which only asks the output to be dense for a part of snapshots.
They investigated the computational complexity of the problems and designed some approximation \rev{and} heuristic algorithms.

Zhu et al. \cite{Zhu2022}
weighted each edge according to the number of time-stamps in which the edge exists and tackled the problem \rev{of maximizing} the density divided by the number of time-stamps in the relative time-interval, by imposing a constraint on the minimum number of time-stamps to take into account. For this problem, the authors designed \rev{exact and heuristic algorithms}.

The most recent contribution in this context is by Qin et al. \cite{qin2023periodic}, which defined a $\sigma$-periodic subgraph to be a subgraph whose occurrences are exactly $\sigma$ in the temporal graph, and the time difference between any pair of consequent occurrences is the same.
They proposed efficient strategies to prune the search space in temporal graphs, in order to run \rev{a max-flow algorithm for} an instance of reduced size.

%% file: sec6/uncertain.tex
\subsection{Uncertain graphs}
Zou~\cite{zou2013polynomial} studied the problem of extracting a \rev{vertex subset that} maximizes the \emph{expected density} from  an uncertain graph, i.e., $G=(V,E,p)$ and $p: E\rightarrow [0,1]$,
where $p(e)$ represents the probability of existence for each $e\in E$ \cite{PotamiasBGK10,BonchiGKV14}. 
\rev{The author} showed that this problem can be reduced to DSP on edge-weighted graphs,
and designed a polynomial-time exact algorithm based on the reduction.
\rev{Recently, Sun et al. \cite{sun2021efficient} applied a probabilistic truss indexing framework to the triangle DSP (see Section~\ref{subsec:numerator}) in uncertain graphs.}

Miyauchi and Takeda~\cite{miyauchi2018robust} considered the uncertainty of edge weights rather than the existence of edges.
To model that, they assumed that there is an edge-weight space $I=\times_{e\in E}[l_e,r_e]\subseteq \times_{e\in E}[0,\infty)$
that contains the unknown true edge weight $w$.
To evaluate the performance of $S\subseteq V$ without any concrete edge weight,
they employed a well-known measure in the field of robust optimization, called the robust ratio.
Intuitively, $S\subseteq V$ with a large robust ratio has a density close to the optimal value even on $G$ with the edge weight selected adversarially from $I$.
Using this, they introduced the robust densest subgraph problem:
given $G=(V,E)$ and $I=\times_{e\in E}[l_e,r_e]\subseteq \times_{e\in E}[0,\infty)$,
we are asked to find $S\subseteq V$ that maximizes the robust ratio under $I$.
They designed an algorithm that returns $S\subseteq V$ with the best possible robust ratio (except for the constant factor), under some mild condition. 
In addition, they also introduced the robust densest subgraph problem with sampling oracle,
where we have access to an oracle that accepts $e\in E$ and outputs a value drawn from a distribution on $[l_e, r_e]$
in which the expected value is equal to the unknown true edge weight, and designed a pseudo-polynomial-time algorithm with a strong quality guarantee.

Tsourakakis et al.~\cite{tsourakakis2019novel} introduced an optimization problem called the risk-averse dense subgraph discovery problem.
Here an uncertain graph is defined as a pair of $G=(V,E)$ and $(g_e(\theta_e))_{e\in E}$,
where the weight $w(e)$ of each edge $e\in E$ is drawn independently from the rest
according to some probability distribution $g_e$ with parameter $\theta_e$.
Each probability distribution $g_e$ is assumed to have finite mean $\mu_e$ and variance $\theta^2_e$.
Intuitively speaking, their risk-averse variant aims to find $S\subseteq V$
that maximizes $\frac{\sum_{e\in E[S]}\mu_e}{|S|}$ but minimizes $\frac{\sum_{e\in E[S]}\sigma^2_e}{|S|}$.
Tsourakakis et al.~\cite{tsourakakis2019novel} showed that this problem reduces to DSP on negatively-weighted graphs,
and designed an efficient approximation algorithm based on the reduction.

Recently, Kuroki et al.~\cite{Kuroki+20} pointed out that
the sampling procedure used in Miyauchi and Takeda~\cite{miyauchi2018robust},
where all edges are repeatedly queried by a sampling oracle that returns an individual edge weight,
is often quite costly or sometimes impossible.
To overcome this issue, \rev{the authors} introduced a novel framework called the densest subgraph bandits (DS bandits),
by incorporating the concept of stochastic combinatorial bandits~\cite{chen2013combinatorial,chen2014combinatorial} into DSP.
In DS bandits, a learner is given $G=(V,E)$ 
whose edge-weights are associated with unknown probability distributions.
During the exploration period, the learner chooses a subset of edges (rather than only single edge, unlike Miyauchi and Takeda~\cite{miyauchi2018robust}) to sample,
and observes the sum of noisy edge weights in a queried subset.
They investigated DS bandits with the objective of best arm identification;
that is, the learner must report one subgraph that (s)he believes to be optimal after the exploration period.
Their first algorithm has an upper bound on the number of samples required to identify a $(1+\epsilon)$-approximate solution
with probability at least $1-\delta$ for $\epsilon>0$ and $\delta \in (0,1)$.
Their second algorithm, based on the greedy peeling algorithm, is scalable and parameter free, 
which has a non-trivial upper bound on the probability that the density of the output is less than half of the optimal value.

The most recent contribution is due to Saha et al. \cite{Saha2022most}, 
that defined DSP in uncertain graphs with an alternative notion of density, called Densest Subgraph Probability (DSPr).
DSPr of $U \subseteq V$ is the summation of the probabilities of all \rev{realizations} in which $U$ represents the densest subgraph (with any \rev{notion of density}).
They designed the Most Probable Densest Subgraph (MPDS) problem considering edge density, clique density, pattern density, and their top-$k$ variants.
Their principal contribution is the edge density-based MPDS algorithm, that is built on independent sampling of possible worlds (e.g., via \rev{a} Monte-Carlo sampling) and \rev{an} efficient enumeration of all edge-densest subgraphs \rev{in each of them}.

%% file: sec6/hypergraphs.tex
\subsection{Hypergraphs}\label{subsec:hyper}
Hypergraphs are a generalization of graphs, consisting of a set $V$ of vertices and a set $E$ of \rev{hyperedges}, that are composed by an arbitrary number of vertices. Huang and Kahng \cite{Huang-Kahng95} formally introduced the Densest Subhypergraph problem (DSH), that given \rev{a} hypergraph $G=(V,E)$, requires to find $S\subseteq V$ that maximizes \rev{$d(S) = e[S]/|S|$, where $e[S]=|\{e\in E\mid e\subseteq S\}|$}, and proposed a \rev{polynomial-time} flow-based \rev{exact} algorithm.

\rev{
Very recently, Huang et al.~\cite{Huang+23} generalized DSH to the model with a seed set, 
inspired by the aforementioned ADS (see Section~\ref{subsec:seed}). 
Their model, called the Anchored Densest Subhypergraph (ADSH), is a common generalization of DSH and ADS, 
while the anchored node set is limited to the emptyset. 
Therefore, the model seeks a densest subhypergraph that is close to the given reference node set. 
Its noteworthy feature is the consideration of the so-called locality parameter, 
specifying how the output should be close to the given reference node set. 
The authors presented a flow-based exact algorithm for ADSH, 
based on that for DSP and the recent development of computing a minimum $s$--$t$ cut on hypergraphs~\cite{Veldt+22}. 
This result remains valid even for some generalizations of ADSH, including HDSP (see Section~\ref{subsec:labeled}), 
on hypergraphs with positive and even negative node weights. 
The main algorithmic contribution is to develop a strongly-local exact algorithm for ADSH, 
i.e., an algorithm outputting an optimal solution in time depending only on the parameters of the reference node set. 
The algorithm runs in an iterative fashion: in the first iteration, the algorithm deals with the subhypergraph consisting of the reference nodes and their neighbors, and the later iterations, it computes the minimum $s$--$t$ cut in a hypergraph constructed from the current subhypergraph and expands it using the information of the cut computed. 
}

\rev{Miyauchi et al. \cite{Miyauchi+15} \rev{studied} two general optimization models, in the context of advertising budget allocation, i.e., 
the maximum general-thresholds coverage problem, in which the densest $k$-subhypergraph 
(i.e., the variant of DSH in which the output size is constrained to be equal to $k$) falls, 
and its cost-effective counterpart, in which DSH itself falls. 
For the first model they proposed two different greedy algorithms, while for the second problem they designed a scalable approximation algorithm.
}

Hu et al. \cite{hu2017dynamicsub} \rev{generalized the results by Tsourakakis~\cite{Tsourakakis15}} to hypergraphs, 
by designing \rev{LP-based and flow-based algorithms}. 
Furthermore, they addressed the densest subhypergraph maintenance in \rev{a} dynamic setting, 
by providing two algorithms that maintain a $r(1+\epsilon)$-approximation in the case where there are only edge insertions, 
and a $r^2(1+\epsilon)$-approximation in \rev{a} fully dynamic setting, where $r$ is the rank of the hypergraph, i.e., $r=\max_{e\in E} |e|$. 
Chlamtac et al. \cite{Chlamtac+18} performed a theoretical analysis on \rev{the} densest $k$-subhypergraph problem and provided bounds over 3-uniform hypergraphs; 
Corinzia et al.~\cite{corinzia2022statistical} considered recovering the planted densest $k$-subhypergraph in a $d$-uniform hypergraph (i.e., \rev{a} hypergraph in which for any $e \in E$, $|e| = d$ holds) \rev{and provided} tight statistical bounds on recovering quality and algorithmic bounds based on approximate message passing algorithms.

\rev{The \rev{recent results} by Chekuri et al.~\cite{Chekuri2022supermod} also have generalization to hypergraphs, e.g., 
a fast $(1+\epsilon)$-approximation algorithm based on max flow.
Furthermore, a natural generalization of {\sc Greedy++} (in Section~\ref{subsubsec:iterative_peeling}) to hypergraphs provides a
$(1+\epsilon)$-approximation for DSH.}
The iterative algorithm proposed by Harb et al.~\cite{harb2022faster} \rev{also has a generalization to hypergraphs}; it provides an $\epsilon$-additive approximate solution for DSH in
$\bigO(\sqrt{pr\Delta} / \epsilon)$ iterations,
where each iteration requires $O(p\log r)$ time and admits some level of parallelization. Here $\Delta=\max_{v \in V} |\{e\in E : v \in e\}|$ and $p = \sum_{e \in E} |e|$.

Zhou et al. \cite{zhou2022extracting} \rev{introduced a generalization of DSH by considering the fact that} there might be \rev{some} hyperedges only partially included in the solution, \rev{which} were not counted in the objective function of DSH but might be relevant in \rev{some} specific applications. Therefore, they \rev{defined} a weighting scheme according to the number of vertices of any hyperedge included in the solution, and \rev{a} maximization problem based on it. For this problem they \rev{proposed} exact and $r$-approximation algorithms, where $r$ is the rank of the hypergraph.

Finally, Bera et al. \cite{bera2022dynamicsub} very recently designed a new algorithm for the dynamic setting. They improved the approximation ratio to $(1+\epsilon)$, making it independent \rev{of} the hypergraph rank, with a similar update time to that required in \cite{hu2017dynamicsub}.

%% file: sec6/metric.tex
\subsection{Metric spaces}\label{subsec:metricscpaces}
A metric space is a pair $(X, d)$, where $X$ is a set and $d: X \times X \rightarrow \mathbb{R}_+$ is a function called a metric, that assigns a non-negative real number, denoted by $d(x,y)$, to any two elements $x, y \in X$, satisfying the following properties:
$d(x, y) = 0 \text{ if and only if } x = y$ (identity), $d(x, y) = d(y, x)$ (symmetry), and $d(x, y) + d(y, z) \geq d(x, z)$ (triangle inequality). We can define a weighted graph in a metric space, where the vertices are the points of the space, and each pair of point is connected by an edge, weighted according to the distance function $d$. As the resulting graph would be complete, sparsification might be needed.
Esfandiari and Mitzenmacher \cite{esfandiari2018metric} designed a novel sampling approach for the edges of a graph defined in a metric space, using a sublinear number of queries, succeeding with probability at least $1-\bigO(1/n)$.
\rev{This} result implies that for DSP in $\lambda$-metric graphs\footnote{In a $\lambda$-metric graph the triangle inequality that holds is $d(x, y) + d(y, z) \geq \lambda d(x, z)$.}, they can provide a $(2+\epsilon)$-approximation algorithm requiring $\tilde{\bigO}(n / (\lambda \epsilon^2))$ time.

%% file: computational.tex
Bahmani et al.~\cite{bahmani} were the first to study DSP in a streaming scenario.
The authors designed algorithms for DSP, DDS (Section \ref{subsec:directed}), and Dal$k$S (Section \ref{sec:size}) under the semi-streaming model of computation, where the set of vertices is known ahead of time and can fit into the main memory, while the edges arrive one by one.
The proposed algorithm is based on the \rev{greedy peeling algorithm},
making $\bigO(\log_{1+\epsilon} n)$ passes over the data \rev{and outputting a $2(1+\epsilon)$-approximate solution} for DSP; the adapted version of the algorithm for the DDS provides the same guarantees.
A slightly modified version of this method gives a $3(1+\epsilon)$-approximation for Dal$k$S, always performing $\bigO(\log_{1+\epsilon} n)$ passes over the input.
The authors also proved that any
$p$-pass streaming
$\alpha$-approximation ($\alpha \geq 2$) algorithm for DSP needs
$\Omega(n / p\alpha^2)$
space.
The authors also demonstrated how the algorithm can be easily parallelized by providing a MapReduce implementation.

Tsourakakis~\cite{Tsourakakis15} exploited these algorithmic techniques and efficient triangle counting algorithms in MapReduce~\cite{suri2011Counting} to address the $k$-clique DSP.
The author gave a
$3(1+\epsilon)$-approximation algorithm
for the \rev{triangle} DSP requiring
$\bigO(\log_{1+\epsilon} n)$ rounds.
Shi \rev{et al.}~\cite{shi2021parallel} addressed DSP under the \textit{workspan model} \cite{jeje1992introduction,cormen2022introduction},
where
the work $W$ of an algorithm is the total number of operations, the span $S$ is the longest dependency path, and $P$ \rev{is the number of processors} available, 
and the time for executing a parallel computation is $\frac{W}{P}+S$.
The authors provided a parallel algorithm that
computes a
$k(1+\epsilon)$-\rev{approximate solution}
for \rev{the} $k$-clique DSP running in
$\bigO(m\alpha^{k-2})$ work \rev{and}
$\bigO(k\log^2 n)$ span w.h.p., \rev{using}
$\bigO(m+P\alpha)$ space, \rev{where} $\alpha$ is the arboricity of the input graph, i.e.,  the minimum number of spanning forests needed to cover the graph.

Das Sarma et al.~\cite{dassarma2012dynamic} were the first to study DSP in a fully decentralized distributed computing peer-to-peer network model, \rev{i.e.,} the CONGEST~\cite{peleg2000distributed} model.
The authors gave a distributed algorithm
that w.h.p. in
$\bigO(D\log_{1+\epsilon} n)$ time provides a
$(2+\epsilon)$-approximation
for DSP and \rev{a} 
$(3+\epsilon)$-approximation
for Dal$k$S, where $D$ is the diameter of the graph.

Bahmani et al. \cite{bahmani2014efficient} presented primal-dual algorithms, working in the MapReduce framework,
that provide a
$(1+\epsilon)$-approximation
for both DSP and DDS, and run in
$\bigO(m(\log n / \epsilon^2))$
time for DSP and $\bigO(m(\log^2 n /\epsilon^3))$ for DDS, by taking
$\bigO(\log n/\epsilon^2)$
MapReduce phases.
The total running time and shuffle size in each phase is $\bigO(m)$ and the reduce-key-complexity is $\bigO(\Delta)$, where $\Delta$ is the maximum degree of a vertex in the graph.
They provided an algorithm for DDS by combining the approach used for DSP with the LP formulation for DDS~\cite{Charikar2000}.

Ghaffari et al.~\cite{ghaffari2019improved} studied DSP under the Massively Parallel Computation (MPC) model, a theoretical abstraction suitable for MapReduce.
The authors provided an algorithm that w.h.p. in
$\bigO(\sqrt{\log n}\,\log \log n)$
rounds computes a
$(1+\epsilon)$-\rev{approximate solution}
for DSP.
The algorithm requires
$\tilde\bigO(n^\delta)$ memory per machine and a total memory of
$\tilde\bigO(\max\{m, n^{1+\delta}\})$ 
for an arbitrary constant
$\delta \in (0,1)$.
Epasto et al.~\cite{epasto2015dynamic} studied DSP in two dynamic graph models.
When edges can be added adversarially, the authors designed an algorithm that maintains at any point in time a
$2(1 + \epsilon)^2$-approximation
of the densest subgraph, performing
$\bigO(m \log^2 n /\epsilon^2)$
operations
and requiring
$\bigO(m+n)$ space. When edges can be added adversarially and removed uniformly at random, the authors proposed an algorithm that maintains at any point in time a
$2(1 + \epsilon)^6$-approximation
of the densest subgraph, performing
$\bigO\left(\frac{A\log A \log^2 n}{\epsilon^2} + \frac{R\log A \log^3 n}{\epsilon^4}\right)$
operations with high probability ($A$ and $R$ are the numbers of edge insertions and removals, respectively) and requiring
$\bigO\left(m+n\right)$ space.
For both algorithms, all results remain valid even in the presence of vertex insertions.
Ahmadian and Haddadan~\cite{ahmadian2021wedge} experimentally showed an improvement in the execution time by embedding this method in their framework at the cost of having a
$4(1 + \epsilon)^2$-approximation for DSP.

In the presence of both edge insertions and deletions,  Bhattacharya et al.~\cite{bhattacharya2015dynamic} designed an algorithm \rev{that provides} a $(4+\epsilon)$-approximation for DSP.
The algorithm requires $\tilde\bigO(n)$
space, \rev{and} it has an amortized time of $\tilde\bigO(1)$ for both insertion and deletion of an edge and a query time of $\tilde\bigO(1)$ for obtaining the declared approximation of the densest subgraph at any point in time.
The authors also \rev{showed} how increasing the query time to
$\tilde\bigO(n)$ yields a $(2+\epsilon)$-approximation for DSP.
When the graph is represented by an incident list, the proposed algorithm
provides a $(2+\epsilon)$-approximation for DSP
by reading $\tilde\bigO(n)$
edges, requiring
$\tilde\bigO(n)$
query time and $\tilde\bigO(n)$ space.
The authors also showed the tightness up to a poly-logarithmic factor of this running time \rev{by} providing a lower bound for the problem.
The authors also demonstrated the algorithm's applicability to the distributed streaming model defined in Cormode et al.~\cite{Cormode2010optimal},
showing that the algorithm computes a $(2+\epsilon)$-approximate solution for DSP using
$\tilde\bigO(k+n)$ bits of communications, where the space required by the coordinator is $\tilde\bigO(n)$ and the space required by each site is $\tilde\bigO(1)$.
The authors also extended the $(4+\epsilon)$-approximation algorithm for DSP to DDS, obtaining a
$(8+\epsilon)$-approximation; the extended algorithm requires $\tilde\bigO(m+n)$ space, an amortized time of $\tilde\bigO(1)$ for both insertion and deletion of an edge, and a query time of $\bigO(1)$.

McGregor et al.~\cite{mcgregor2015dynamic} \rev{provided} a single-pass algorithm that returns a $(1+\epsilon)$-approximation for DSP with high probability; their algorithm uses
$\tilde\bigO(n)$
space,
$\tilde\bigO(1)$
time for each edge insertion and deletion, and
$\bigO(\poly n)$
query time.
The space used by the algorithm matches the lower bound defined in \cite{bahmani} up to a poly-logarithmic factor for constant $\epsilon$. At the same time, the polynomial query time follows by using any exact algorithm for DSP on the subgraph generated by the algorithm.

Esfandiari et al.~\cite{esfandiari2016dynamic} also provided a single-pass algorithm that returns a $(1+\epsilon)$-approximation for DSP with high probability, 
using
$\tilde\bigO(m)$ space,
$\tilde\bigO(1)$
time for each edge insertion and deletion, and
$\bigO(\poly n)$
query time.
Their algorithm is based on sampling (and storing) edges from the input graph uniformly at random without replacement.
As in McGregor et al.~\cite{mcgregor2015dynamic}, even in this case polynomial query time follows by using any exact algorithm for DSP on the subgraph generated by the algorithm.
The authors also proved that using the greedy peeling algorithm at query time,
and using an oracle that provides direct access to a uniformly sampled edge,
the proposed algorithm gives a
$(2+\epsilon)$-approximation for DSP with high probability
running in
$\tilde\bigO(n)$ time.
The authors proved that the generality of their algorithmic framework provides a
$(1+\epsilon)$-approximation for DDS with high probability, 
using
$\tilde\bigO(n^{\frac{3}{2}})$ space,
$\tilde\bigO(1)$
time for each edge insertion and deletion, and
$\bigO(\poly n)$
query time for solving DDS exactly.

Su and Vu~\cite{su2020distributed} gave a distributed algorithm under the CONGEST~\cite{peleg2000distributed} model that w.h.p. in
$\bigO(D+(\log^4 n / \epsilon^4))$ time provides a
$(1+\epsilon)$-approximation
for DSP, where $D$ is the diameter of the graph.
Sawlani and Wang~\cite{sawlani2020dynamic}
gave a deterministic algorithm providing a
$(1+\epsilon)$-approximation for DSP,
with
$\bigO(\log^4 n / \epsilon^6)$
worst-case time per edge insertion or deletion, and
$\bigO(1)$ worst-case query time.
The \rev{authors also gave} a deterministic algorithm providing a
$(1+\epsilon)$-approximation for DDS,
with
$\bigO(\log^5 n / \epsilon^7)$
worst-case time per edge insertion or deletion, and
$\bigO(\log n / \epsilon)$ worst-case query time.
Moreover, for both DSP and DDS, at any point in time, the corresponding algorithms can output the approximate \rev{solution $S$} in 
$\bigO(|S| + \log n)$ time. 
Christiansen et al.~\cite{christiansen2022adaptive}  improved the result of Sawlani and Wang~\cite{sawlani2020dynamic}
\rev{by} giving a deterministic algorithm providing a
$(1+\epsilon)$-approximation for DSP,
with
$\bigO((\log^2 n \log \rho^*) / \epsilon^4)$
worst-case time per edge insertion or deletion,
where $\rho^*$ is the density of the densest subgraph, and
$\bigO(1)$ worst-case query time.
At any point in time, the algorithms can output the approximate \rev{solution $S$} in 
$\bigO(|S|)$ time. 
Independently from Christiansen et al.~\cite{christiansen2022adaptive}, Chekuri and Quanrud~\cite{chekuri2022dynamic} also improved the result of Sawlani and Wang~\cite{sawlani2020dynamic}
\rev{by} giving a deterministic algorithm providing a
$(1+\epsilon)$-approximation for DSP,
with
$\bigO(\log^2 n / \epsilon^4)$ amortized time and
$\bigO((\log^3 n \log\log n) / \epsilon^6)$
worst-case time per edge insertion or deletion,
and
$\bigO(1)$ worst-case query time.
At any point in time, the \rev{algorithm} can output the approximate solution $S$ in 
$\bigO(|S|)$ time. 
\rev{Building upon this work, Li and Quanrud~\cite{li2023approximate} recently proposed a fully dynamic algorithm that maintains a 
$(1+\epsilon)$-approximate 
solution for DDS in 
$\bigO((\log^3 n \log\log n)/\epsilon^6)$ 
amortized time or 
$\bigO((\log^4 n \log\log n )/\epsilon^7)$ worst-case time per edge insertion or deletion.
This is a slight improvement over Sawlani and Wang~\cite{sawlani2020dynamic}, when ignoring the $\log\log$ factors.}

Henzinger et al.~\cite{henzinger2022fine}
showed that for DSP
on graphs with a maximum degree less than or equal to seven, constant-degree graphs, expanders, and power-law graphs
there is no dynamic algorithm \rev{that} achieves both
$\bigO(n^{1/4-\epsilon})$ amortized edge update time and $\bigO(n^{1/2-\epsilon})$ amortized query time.
This conditional lower bound is weaker than the one for general graphs:  
$\bigO(n^{1/2-\epsilon})$ amortized edge update time and $\bigO(n^{1-\epsilon})$ amortized query time.

Chu et al.~\cite{chu2022hierarchical} presented a parallel algorithm that in time
$\bigO\left( n\sqrt{p} + m\alpha(n) + F\right)$
computes a $2$-approximate solution for DSP requiring $\bigO\left(n\right)$ additional space to the input storage, where $p$ is the number of threads, $F$ is the number of failures, and $\alpha(n)$ is the inverse of the Ackermann function.
The authors stated that for practical input, the dominant term is
$m\alpha(n) \leq 4m$;
consequently, the time complexity reduces to
$\bigO(m)$.

The first contribution tackling DSP on hypergraphs in a streaming scenario was by Hu~et~al.~\cite{hu2017dynamicsub}, which developed two dynamic algorithms for DSH (see Section \ref{subsec:hyper})  by extending the pioneering algorithm of~\cite{bahmani} to hypergraphs.
With only arbitrary edge insertions, the authors gave a dynamic algorithm providing a
$r(1+\epsilon)$-approximation for DSH, where $r$ is the rank of a hypergraph,
with
$\bigO\left(\poly\left(\frac{r}{\epsilon} \log n \right)\right)$ amortized time per edge insertion and $\bigO(n)$ extra space in addition to the input hypergraph.
When both arbitrary edge insertions and deletions are allowed, the authors gave a dynamic algorithm providing a
$r^2(1+\epsilon)$-approximation for DSH,
with
$\bigO\left(\poly\left(\frac{r}{\epsilon} \log n \right)\right)$ amortized time for each edge insertion and deletion and $\bigO\left(r m \poly\left(\frac{r}{\epsilon} \log n \right)\right)$ extra space in addition to the input hypergraph.

Bera et al.~\cite{bera2022dynamicsub} addressed DSH on \rev{weighted hypergraphs} in a streaming scenario \rev{and gave}
a randomized fully dynamic algorithm providing a
$(1+\epsilon)$-approximation for the problem.
Chekuri and Quanrud~\cite{chekuri2022dynamic} presented a deterministic algorithm providing a
$(1+\epsilon)$-approximation for DSH,
with
$\bigO\left(\frac{r^2 \log^2 n}{\epsilon^4}\right)$ amortized time and
$\bigO\left( \frac{r\log(n) \log(\rho^*)}{\epsilon^4} + \frac{r^2\log^3(n) \left(\log\log(n) + \log (\epsilon^{-1}) \right)}{\epsilon^6} \right)$
worst-case time per hyperedge insertion or deletion (where $\rho^*$ is the density of the densest subhypergraph),
and
$\bigO(1)$ worst-case query time.
At any point in time, the algorithm can output the approximate solution $S$ in $\bigO(|S|)$ time.

%% file: apps.tex
In this section we provide a brief and non-exhaustive, yet representative, coverage of
application domains in which DSP paved the way to interesting solutions for real-world problems.

\spara{Web and Social Networks.} Densest Subgraph has been usefully exploited in developing solutions for several problems related to the web and to social networks. 
\rev{Gibson et al.~\cite{gibson} and Dourisboure et al.~\cite{Dourisboure+07,DourisboureGP09} demonstrated that dense subgraph discovery algorithms are useful for extracting communities in the web graphs.}
\rev{Gajewar and Das Sarma~\cite{Gajewar+12} considered the problem of identifying a team of skilled individuals for collaboration, in which the goal is to maximize the collaborative compatibility of the team, and formulated it as a variant of DSP. 
}
Rangapuram et al. \cite{rangapuram2013team} \rev{also} addressed the team formation problem, in which there is need of selecting a set of employees to complete a certain task, under the constraint \rev{on} a cost function.
Angel et al.~\cite{Angel2013} studied real-time story identification on Twitter via maintenance of \rev{the} densest subgraph in the fully dynamic setting.
Similarly, Bonchi et al. \cite{BonchiBGS16,BonchiBGS19} adopted anomalously dense subgraphs in temporal networks as a way to identify buzzing stories in social media.
Hooi et al. \cite{hooi2016fraudar} approached fraudulent reviews detection via dense subgraph discovery in user-product bipartite graphs.
Kawase et al.~\cite{kawase2019crowd} improved the extraction of reliable experts for answer aggregation in the crowdsourcing framework.
Yikun et al.~\cite{yikun2019no} devised an algorithm for detecting fraudulent entities in tensors based on densest subgraphs.
Kim et al.~\cite{kim2020densely} modeled the community search problem taking into account spatial locations.
Tan et al.~\cite{tan2020scaling} proposed to employ the densest subgraph for candidate committers selection in open source communities.
Fazzone et al. \cite{Fazzone2022} adopted HDSP (see Section \ref{subsec:labeled}) to detect polarized niches in social networks, i.e., set of users that are far from authoritative sources of information and at the same time close to misinformation spreaders.

\spara{Biology.} In biology and in particular in ``omics'' disciplines, there are plenty of settings in which the data is represented with graphs \cite{koutrouli2020guide}. For instance, gene co-expression networks or Protein-Protein Interaction \rev{(PPI)} networks, are graphs built from correlation matrices: in such networks a dense subgraph can represent a set of genes/proteins that are regulating the same process.
Hu et al. \cite{hu2005coherent} defined algorithms for dense and coherent subgraphs across networks, in order to detect recurrent patterns across multiple networks to discover biological modules.
Fratkin et al. \cite{fratkin} exploited dense subgraphs to find regulatory motifs in genomic DNA, by creating a graph where vertices correspond to $k$-mers (sequence of $k$ DNA bases) and edges to $k$-mers that differ in few positions.
Everett et al. \cite{Everett} proposed to extract dense structures in networks composed by transcription factors, their putative target genes, and the tissues in which the target genes are differentially expressed; in this context, they defined a dense subgraph as a transcriptional module.
Saha et al. \cite{SahaHKRZ10} studied the problem of finding complex annotation patterns in gene annotation graphs. Given a distance metric between any pair of nodes, the densest subgraph respecting a specific minimum distance threshold is produced.
Feng et al. \cite{feng} combined the information brought by \rev{PPI} data and microarray gene expression profiles, and provided a densest-subgraph-based algorithm to identify protein complexes.
Li et al. \cite{Li2022densestbio} recently tested different algorithms for DSP (and related problems) for detecting hot spots in \rev{PPI} networks.
Lanciano et al. \cite{lanciano2022biocontrast} modeled the differential co-expression analysis problem with the detection of contrastive subgraphs in co-expression networks of different subtype of breast cancer.
Martini et al.~\cite{martini2022network} developed a dense subgraph searching method for jointly prioritizing putative causal genes for disease and selecting one biologically similar potential causal gene at each genetic risk locus.

\spara{Finance.} Another domain in which \rev{it} is natural to model the data with graphs is finance. Boginski et al.~\cite{BBP} exploited dense subgraphs to predict the behavior of financial instruments, through the lens of the maximum clique. In fact, correlating the stock trends and representing these with a graph, is possible to define a dense structure as a set of stocks whose trend is similar, and viceversa for an independent set.
Li et al. \cite{li2020money} proposed to detect money laundering, modeling the transactions with a multipartite directed graph.
Ren et al.~\cite{ren2021ensemfdet} employed densest subgraphs to design an ensemble method for fraud detection in e-commerce.
Jiang et al.~\cite{Jiang2022spade} developed a real-time fraud detection framework called Spade, which can detect fraudulent communities in hundreds of microseconds on million-scale evolving graphs by incrementally maintaining dense subgraphs.
Chen and Tsourakakis \cite{chen2022anti} approached fraud detection in financial networks by mining dense subgraphs deviating significantly from Benford's law, which describes the distribution of the first digit of numbers appearing in a wide variety of numerical data, and has been used to raise ``red flags'' about potential anomalies in the data such as tax evasion.
Ji et al. \cite{cash_out_2022} proposed to identify cash-out \rev{behaviors}, i.e., withdrawal of cash from a credit card by illegitimate payments with merchants, with densest subgraphs subject to the optimization of a class of suspiciousness metrics.
Xie et al.~\cite{xie2022orion} applied DSP algorithms inside their algorithmic proposal for zero-knowledge proof, a powerful cryptographic primitive that has found various applications.

\spara{Miscellanea.}
\rev{Chen and Saad~\cite{Chen+12} employed a dense subgraph discovery algorithm in the scenario of community detection, where the number of communities is unknown and some vertices may not belong to any of them.} 
Moro et al.~\cite{Moro2014entity} devised an algorithm for entity linking and word sense disambiguation, using densest subgraphs.
Rozenshtein et al.~\cite{rozenshtein2014event} detected interesting events in activity networks, by maximizing the density minus the distance between all the nodes included in the final solution. 
Different works also modeled the reachability and distance query indexing problem via the densest subgraph framework~\cite{CohenHKZ02,JinXRF09}. 
\rev{Shin et al.~\cite{Shin+16,Shin+17,Shin+17_2} developed efficient algorithms for detecting dense subtensors based on densest subgraphs.}
Kamara and Moataz~\cite{kamara2019computationally} implemented structured encryption schemes with computationally-secure leakage, based on the hardness results of the planted densest subgraph problem.
Lanciano et al.~\cite{lanciano2020contrast} proposed to detect the most contrastive subgraph in terms of density for different groups of brain networks.
Wu et al.~\cite{wu2021extracting} extracted densest subgraphs in brain networks, proposing a likelihood-based objective function, to identify brain regions associated to schizophrenia disorder.
Majbouri et al.~\cite{majbouri2020prediction} leveraged DSP to boost the prediction of information diffusion paths in social networks, while Bhadra and Bandyopadhyay~\cite{bhadra2021supervised} exploited it to perform a better feature selection.
Yan et al.~\cite{yan2021anomaly} performed anomaly detection of network streams via densest subgraphs.
DSP on vertex-weighted graphs is used for efficient distribution of quantum circuits by Sundaram et al.~\cite{sundaram2021efficient}.
Lusk et al.~\cite{lusk2021clipper} designed a generalized version of the maximum clique problem to perform robust data association.
Konar and Sidiropoulos~\cite{konar2022triangle} showed that the generalization TD$k$S of D$k$S is useful for unsupervised document summarization.
Sukeda \rev{et al.}~\cite{sukeda2022study} accelerated a column generation algorithm for a clustering problem called the modularity density maximization problem, using the greedy peeling algorithm for a variant of DSP.
Recently Chen et al.~\cite{chen2022algorithmic} introduced a novel framework for motif detection, whose algorithmic proposal for testing the statistical significance of a single motif is based on the greedy peeling for DSP. 
\rev{
Very recently, Ding and Du~\cite{Ding+23} addressed the problem of detecting the edge correlation between a pair of Erd\H{o}s--R\'enyi graphs, based on the observation that the detection problem is related to DSP. 
}

%% file: conclusions.tex
The Densest Subgraph Problem is a classic and fundamental problem in graph theory that has received a great deal of attention. In the last couple of years we have witnessed a renewed interest in the problem due to its relevance in emerging applications, which has led to several algorithmic breakthroughs. This revival of interest motivates the present survey, which offers a comprehensive overview of the Densest Subgraph Problem and its many variants, covering a long and rich literature spanning over five decades.

We have summarized the different techniques and algorithms used to solve the problem in its classical setup, including exact and approximation algorithms, and presented a detailed collection of natural variants that have emerged over the years, either with respect to the definition of the objective function, the introduction of different constraints or the typology of the input graph.
Derived from the classical definition of DSP, these variants possess various algorithmic properties that are advantageous, and they also share the elegant and intuitive building foundations of the original problem.
\rev{Our extensive investigations of both models and algorithms provide} a valuable resource for researchers seeking to solve related problems to the Densest Subgraph.
We hope that our survey will inspire new and innovative applications of the Densest Subgraph Problem, as well as serve as a useful reference for researchers in the field. While we acknowledge that gaps may exist in our coverage, we believe that the comprehensive nature of our review will be a valuable asset for future research in this area.

\spara{Open Problems.} Recent advancements have significantly enhanced the efficiency of Densest Subgraph extraction in its classical form, bringing the problem close to being fully resolved. However, several important open problems still require attention, and it is crucial that the scientific community continues to work towards addressing them.

An essential milestone for future research in Densest Subgraph mining is to establish the exact convergence of the solution provided by the iterative peeling method introduced in \cite{boob2020flowless} (Section \ref{subsubsec:iterative_peeling}). Defining this convergence would confirm the existence of a highly efficient method with strong guarantees for computing Densest Subgraphs in any existing graph.
The iterative peeling method has already shown tremendous promise in delivering excellent results and fast computation times. However, a clear understanding of its exact convergence properties is necessary to assess its performance and potential limitations fully.

There is still room for enhancing constrained versions of DSP, including D$k$S and Dal$k$S, which currently have approximation ratios of $O(n^{1/4+\epsilon})$ and $2$, respectively (Section \ref{sec:cons}). Improving these ratios would represent a significant advancement in the field, making it possible to solve these complex problems more accurately and efficiently.

Distance-based generalization of \rev{the} degree opens the door to many interesting problems. The approximation of the distance-$h$ densest subgraph (Section \ref{subsec:numerator}) by means of distance-generalized core decomposition, is only a first step which leaves wide open space for improvements.

While significant advancements have been made in negatively weighted graphs (Section \ref{sec:negative_weights}), the current approximation ratio is unsatisfactory as it depends \rev{heavily} on the input graphs. Improving this ratio would greatly benefit various applications that require this type of input, such as graphs derived from correlation matrices, and related problems like exclusion queries \cite{tsourakakis2019novel}.

Ensuring fairness \rev{and diversity} is a critical aspect that researchers are increasingly addressing, as it is fundamental to achieving equitable outcomes and combating algorithmic bias. However, the current results in this area are limited \rev{to \cite{anagnostopoulos2020fair,Miyauchi+23}}, underscoring the urgent need for the research community to focus more attention on this issue and contribute to improving upon these findings.

The literature on graphs defined on metric spaces has received relatively little attention (Section \ref{subsec:metricscpaces}), which leaves room for further improvement. These types of graphs are important in applications that rely on spatial data, such as image and video processing, transportation networks, and wireless sensor networks. Algorithms designed to work with these graphs can be crucial for advancing research in these fields.